\documentclass[letterpaper,11pt]{article}
\usepackage{amsmath, amsthm, amssymb}
\usepackage[usenames, dvipsnames]{color}
\usepackage[normalem]{ulem} 
\usepackage{fullpage}
\usepackage[numbers, sort]{natbib}
\usepackage[noend]{algorithmic}
\usepackage{tikz, pgfplots}
\usetikzlibrary{shapes.gates.logic.US,trees,positioning,arrows}
\usepackage{enumerate}
\usepackage{thmtools} 
\usepackage{multirow}
\usepackage{comment}
\usepackage[colorlinks=true,allcolors=blue]{hyperref}
\usepackage{xspace,color}
\usepackage{graphicx}
\usepackage{caption}
\usepackage{subcaption}
\usepackage{enumitem,linegoal}
\usepackage[linesnumbered,ruled,vlined]{algorithm2e}
\usepackage{comment}	
\usepackage{ctable}					
\newtheorem{claim}{Claim}[section]

\usepackage{mathtools}
\usepackage[flushleft]{threeparttable}
\usepackage{verbatim}

\usepackage{footnote}
\makesavenoteenv{tabular}
\makesavenoteenv{table}
\usepackage{relsize}
\usepackage{thm-restate}

\usepackage{xspace}

\usepackage[margin=1in]{geometry}

\usepackage{hyperref}
\newcommand*{\email}[1]{\href{mailto:#1}{\nolinkurl{#1}} }

\definecolor{crimsonglory}{rgb}{0,0,0}

\definecolor{commentcolor}{rgb}{0.75, 0.0, 0.2}

 \newtheorem{theorem}{Theorem}[section]
 \newtheorem{lemma}[theorem]{Lemma}
 
 \newtheorem{corollary}[theorem]{Corollary}
 
 \newtheorem{definition}[theorem]{Definition}

\newif\ifqed
 


\usetikzlibrary{arrows,shapes,snakes,automata,backgrounds,petri,calc}
\usepackage[latin1]{inputenc}

\newtheoremstyle{named}{}{}{\itshape}{}{\bfseries}{.}{.5em}{\thmnote{#3}}
\theoremstyle{named}
\newtheorem*{namedtheorem}{Theorem}

\newcounter{proccnt}

\newcommand{\konote}[1]{}

\title{Distributed Algorithms for Matching in Hypergraphs}

\author{
	Oussama Hanguir\\
	\email{oh2204@columbia.edu}\\
	Columbia University
	\and Clifford Stein\thanks{Research partially supported by NSF grants CCF-1714818 and CCF-1822809.}\\
	\email{cs2035@columbia.edu}\\
	Columbia University
}

\begin{document}
	\newcommand{\ignore}[1]{}
\sloppy
\date{}

\maketitle

\thispagestyle{empty}

\begin{abstract}
We study the $d$-Uniform Hypergraph Matching ($d$-UHM) problem: given an $n$-vertex hypergraph $G$ where every hyperedge is of size $d$, find a maximum cardinality set of disjoint hyperedges. For $d\geq3$, the problem of finding the maximum matching is
$\mathcal{NP}$-complete, and was one of Karp's 21 $\mathcal{NP}$-complete problems. 
In this paper we are interested in the problem of finding matchings in hypergraphs in the massively
parallel computation (MPC) model that is a common abstraction of MapReduce-style
computation. In this model, we present the first three parallel algorithms for $d$-Uniform Hypergraph Matching, and we analyse them in terms of resources such as memory usage, rounds of communication needed, and approximation ratio. The highlights include:\begin{itemize}
    \item A $O(\log n)$-round $d$-approximation algorithm that uses $O(nd)$ space per machine.
    \item A $3$-round, $O(d^2)$-approximation algorithm  that uses $\tilde{O}(\sqrt{nm})$ space per machine.
    \item A $3$-round algorithm that computes a subgraph containing a $(d-1+\frac{1}{d})^2$-approximation, using $\Tilde{O}(\sqrt{nm})$ space per machine for linear hypergraphs, and $\Tilde{O}(n\sqrt{nm})$ in general. 
\end{itemize} 
For the third algorithm, we introduce the concept of HyperEdge Degree Constrained Subgraph (HEDCS), which can be of independent interest. We show that an HEDCS contains a fractional matching with total value at least $|M^*|/(d-1+\frac{1}{d})$, where $|M^*|$ is the size of the maximum matching in the hypergraph. Moreover, we investigate the experimental performance of these algorithms both on random input and real instances. Our results support the theoretical bounds and  confirm the trade-offs between the quality of approximation and the speed of the algorithms.
\end{abstract}

\newpage

\section{Introduction}\label{introduction}

As massive graphs become more ubiquitous, the need for scalable  parallel and distributed algorithms
that solve graph problems grows as well. In recent years, we have seen progress in many graph problems (e.g. spanning trees,
connectivity, shortest paths \cite{andoni2014parallel,andoni2018parallel}) and, most relevant to this work, matchings \cite{ghaffari2018improved, czumaj2019round}.  
 A natural generalization of matchings in graphs is to matchings in {\em hypergraphs}. 
 Hypergraph Matching is an important problem with many  applications such as capital budgeting, crew scheduling,
facility location, scheduling airline flights \cite{skiena1998algorithm}, forming a coalition structure in multi-agent systems \cite{sandholm1999coalition} and determining the winners in combinatorial auctions 
\cite{sandholm2002algorithm} (see \cite{vemuganti1998applications} for a partial survey). Although matching
 problems in graphs are one of the most well-studied problems in algorithms and optimization, the NP-hard problem of finding a maximum matching
 in a hypergraph is not as well understood.
 
 In this work, we are interested in the problem of finding matchings in very large hypergraphs, large enough that we cannot solve the problem on one computer. We develop, analyze and experimentally evaluate three parallel algorithms for hypergraph matchings in the MPC model.  Two of the algorithms are generalizations  of parallel algorithms for matchings in graphs.  The third algorithm develops new machinery which we call a hyper-edge degree constrained subgraph (HEDCS), generalizing the notion of an edge-degree constrained subgraph (EDCS).  The EDCS has been recently used in parallel and dynamic algorithms for graph matching problems \cite{bernstein2015fully, bernstein2016faster, assadi2019coresets}.  We will show a range of algorithm tradeoffs between approximation ratio, rounds, memory and computation, evaluated both as worst case bounds, and via computational experiments.  
 
 More formally,
a {\em hypergraph} $G$ is a pair $G = (V,E)$ where $V$ is the set of vertices and
$E$ is the set of hyperedges. A {\em hyperedge} $e \in E$ is a nonempty subset of the vertices. The cardinality of a hyperedge is the number of vertices it contains. When every
hyperedge has the same cardinality $d$, the hypergraph is said to be $d$-\textit{uniform}. 
A hypergraph is \textit{linear} if the intersection of any two hyperedges has at most one vertex.
A {\em hypergraph matching} is a subset of the hyperedges $M \subseteq E$ such that every
vertex is covered at most once, i.e. the hyperedges are mutually disjoint. This notion generalizes matchings in graphs. 
The
cardinality of a matching is the number of hyperedges it contains. 
A matching is called maximum if it has the largest cardinality of all possible matchings, and  maximal if it is not contained in any other matching.
In the $d$-Uniform Hypergraph Matching Problem
(also referred to as Set Packing or $d$-Set Packing), a $d$-uniform hypergraph is given and one needs to find the maximum cardinality matching. 

 We adopt the most restrictive MapReduce-like model of modern parallel computation among \cite{karloff2010model, goodrich2011sorting, beame2013communication, andoni2014parallel}, the Massively Parallel Computation (MPC) model of \cite{beame2013communication}. This model is widely used to solve different graph problems such as matching, vertex cover  \cite{lattanzi2011filtering, ahn2018access, assadi2017randomized, assadi2019coresets, ghaffari2018improved}, independent set \cite{ghaffari2018improved, harvey2018greedy}, as well as many other algorithmic problems. In this model,  we have $k$ machines (processors) each with space $s$.  $N$ is the size of the input and our algorithms will satisfy
$k\cdot s = 
\tilde{O}(N)$, which means that the total space in the system is only a polylogarithmic factor more than the input size.
The computation proceeds in rounds. At
the beginning of each round, the data (e.g. vertices and edges) is distributed across the machines. In each round, a machine performs local computation on its data
(of size $s$), and then sends messages to other machines for the next round. Crucially, the total amount of
communication sent or received by a machine is bounded by $s$, its space.
For example, a machine can send one message of size $s$, or $s$ messages of size 1. It cannot, however, broadcast a size $s$ message to every machine. 
Each machine treats the received messages as the input for the next round.  Our model limits
the number of machines and the memory per machine
to be substantially sublinear in the size of the input.
On the other hand, no restrictions are placed on
the computational power of any individual machine. The main complexity measure is therefore the memory per machine and the number of rounds $R$ required to solve a problem, which we consider
to be the ``parallel time'' of the algorithm. For the rest of the paper, $G(V, E)$ is a $d$-uniform hypergraph with $n$ vertices and $m$ hyperedges, and when the context is clear, we will simply refer to $G$ as a graph and to its hyperedges as edges. $MM(G)$ denotes the
maximum matching in $G$, and $\mu(G) = |MM(G)|$ is the size of that matching. We define
$d_G(v)$ to be the degree  of a vertex $v$ (the number of hyperedges that $v$ belongs to) in $G$. When the context is clear, we will omit indexing by $G$ and simply denote it $d(v)$.

\section{Our contribution and results.} We design and implement algorithms for the $d$-UHM in the MPC model. We will give three different algorithms, demonstrating different trade-offs between the model's parameters.  
Our algorithms are inspired by methods to find maximum matchings in graphs, but require developing significant new tools to address hypergraphs. We are not aware of previous algorithms for hypergraph matching in the MPC model. First we generalize the randomized coreset algorithm of \cite{assadi2017randomized}  which finds an 3-rounds $O(1)$-approximation for matching in graphs. Our algorithm partitions the graph into
random pieces across the machines, and then simply picks a maximum matching of each machine's subgraph. We show that this
natural approach results in a $O(d^2)$-approximation. While the algorithmic generalization is straightforward, the analysis requires several new ideas.

\begin{namedtheorem}[Theorem \ref{result1} (restated)]
There exists an $MPC$ algorithm that with high probability
computes a $(3d(d-1)+3+\epsilon)$-approximation for the $d$-UHM problem in 3 $MPC$ rounds on machines of memory $s = \Tilde{O}(\sqrt{nm})$.
\end{namedtheorem}

Our second result concerns the MPC model with per-machine memory $O(d\cdot n)$. We adapt the sampling technique and post-processing strategy of \cite{lattanzi2011filtering} to construct maximal matchings in hypergraphs, and are able to show that in $d$-uniform hypergraphs, this technique yields a maximal matching, and thus a $d$-approximation to the $d$-UHM problem in $O(\log{n})$ rounds.

\begin{namedtheorem}[Theorem \ref{mpcmax} (restated)]
There exists an $MPC$ algorithm that given a $d$-uniform hypergraph $G(V, E)$ with high probability
computes a maximal matching in $G$ in $O(\log{n})$ $MPC$ rounds on machines of memory $s =\Theta(d\cdot n)$.
\end{namedtheorem}

Our third result generalizes  the edge degree constrained subgraphs (EDCS), originally introduced by Bernstein and Stein \cite{bernstein2015fully} for maintaining large matchings in dynamic
graphs, and later used for several other problems including matching and vertex cover in the streaming and MPC model \cite{assadi2019coresets, bernstein2016faster}.
We call these generalized subgraphs hyper-edge degree constrained subgraphs (HEDCS). We show that they exist for specific parameters, and that they contain a good approximation for the $d$-UHM problem. We prove that an HEDCS of a hypergraph $G$, with well chosen parameters, contain a fractional matching with a value at least $\frac{d}{d^2-d+1}\mu(G)$, and that the underlying fractional matching is special in the sense that for each hyperedge, it either assigns a value of $1$ or a value less than some chosen $\epsilon$. We call such a fractional matching an $\epsilon$-restricted fractional matching. 
For Theorem \ref{result4}, we compute an HEDCS of the hypergraph in a distributed fashion. This procedure relies on the
robustness properties that we prove for the HEDCSs under sampling.
\begin{namedtheorem}[Theorem \ref{mainth} (restated informally)]
Let $G$ be a $d$-uniform hypergraph and $0<\epsilon <1$. There exists an HEDCS that contains an $\epsilon$-restricted fractional matching $M^{H}_f$ with total value at least $\mu(G)\left(\frac{d}{d^2-d+1} - \epsilon\right)$.
\end{namedtheorem}

\begin{namedtheorem}[Theorem \ref{result4} (restated)]
There exists an $MPC$ algorithm that given a $d$-uniform hypergraph $G(V, E)$, where $|V|= n$ and $|E| = m$, can construct an HEDCS of $G$ in 2 $MPC$ rounds on machines of memory $s = \tilde{O}(n\sqrt{nm})$ in general and $s = \tilde{O}(\sqrt{nm})$ for linear hypergraphs.
\end{namedtheorem}

\begin{namedtheorem}[Corollary \ref{result5} (restated)]
There exists an $MPC$ algorithm that with high probability achieves a $d(d-1+1/d)^2$-approximation to the $d$-Uniform Hypergraph Matching in 3 rounds.
\end{namedtheorem}

Table \ref{table:results} summarizes our results.
{\small
\begin{center}

\begin{tabular}{|c|c|c|c|}
  \hline
  Approximation ratio & Rounds & Memory per machine & Computation per round  \\
  \hline
  $3d(d-1) + 3$ & 3 & $\Tilde{O}(\sqrt{nm})$ & Exponential \\
  \hline
  $d$ & $O(\log{n})$ & $O(dn)$ & Polynomial \\
  \hline
   $d(d-1+1/d)^2$ & 3 & $\Tilde{O}(n\sqrt{nm})$ in general & Polynomial \\
   & & $\Tilde{O}(\sqrt{nm})$ for linear hypergraphs &  \\
  \hline
\end{tabular}
$ $
\captionof{table}{Our parallel algorithms for the $d$-uniform hypergraph matching problem.} 
\label{table:results}
\end{center}
}

\noindent\textbf{Experimental results.} We implement our algorithms in a simulated MPC model environment and test them both on random and real-world instances. Our experimental results are consistent with the theoretical bounds on most instances, and show that there is a trade-off between the extent to which the algorithms use the power of parallelism and the quality of the approximations. This trade-off is illustrated by comparing the number of rounds and the performance of the algorithms on machines with the same memory size. See Section~\ref{sec:experiments} for more details.\\

\noindent\textbf{Our techniques.} For Theorem \ref{result1} and Theorem \ref{result4}, we use the concept of composable
coresets, which has been employed in several distributed optimization models such as the streaming and MapReduce models \cite{abbar2013diverse,mirzasoleiman2013distributed,badanidiyuru2014streaming,balcan2013distributed,bateni2014mapping,indyk2014composable}. Roughly speaking, the main idea behind this
technique is as follows: first partition the data into smaller parts. Then compute a representative solution,
referred to as a coreset, from each part. Finally, obtain a solution by solving the optimization problem
over the union of coresets for all parts. We use a randomized variant of composable coresets, first introduced in \cite{mirrokni2015randomized}, where
the above idea is applied on a random clustering of the data. This randomized variant has been used after for Graph Matching and Vertex Cover \cite{assadi2017randomized, assadi2019coresets}, as well as Column Subset Selection \cite{bhaskara2016greedy}.
Our algorithm for Theorem \ref{result1} is similar to previous works, but the analysis requires new techniques for handling hypergraphs. Theorem \ref{mpcmax} is a relatively straightforward generalization of the corresponding result for matching in graphs.

The majority of our technical innovation is contained in Theorem \ref{mainth} and \ref{result4}.
Our general approach is to construct, in parallel, an HEDCS that will contain a good approximation of the maximum matching in the original graph and that will fit on one machine. Then we can run an approximation algorithm on the resultant HEDCS to come up with a good approximation to the maximum matching in this HEDCS and hence in the original graph.
In order to make this approach work, we need to generalize much of the known EDCS machinery \cite{bernstein2015fully,bernstein2016faster} to hypergraphs.  This endeavor is quite involved, as almost all the proofs do not generalize easily and, as the results show, the resulting bounds are weaker than those for graphs. We first show that HEDCSs exist, and that they contain large fractional matchings.  We then  show that they are useful as a coreset, which amounts to showing that even though
there can be many different HEDCSs of some fixed hypergraph $G(V, E)$, the degree distributions of every HEDCS (for
the same parameters) are almost identical. 
In other words, the degree of any vertex $v$ is
almost the same in every HEDCS of $G$. 
We show also that HEDCS are robust under edge sampling, in the sense that edge sampling from a HEDCS yields another HEDCS. 
These properties allow to use HEDCS in our
coresets and parallel algorithm in the rest of the paper.



\section{Related Work}
\textbf{Hypergraph Matching.} The problem of finding a maximum matching in $d$-uniform hypergraphs is NP-hard for any $d \geq 3$ \cite{karp1972reducibility, papadimitriou2003computational}, and APX-hard for $d=3$ \cite{kann1991maximum}.  The most
natural approach for an approximation algorithm for the hypergraph matching problem is the greedy algorithm: repeatedly add an edge that doesn't intersect any edges already in the matching.
This solution is clearly within a factor of $d$ from optimal: from the
edges removed in each iteration, the optimal solution can contain at most $d$
edges (at most one for each element of the chosen edge). It is also easy to construct examples showing that this analysis is tight for the greedy approach.
All the best known approximation algorithms for the Hypergraph
Matching Problem in $d$-uniform hypergraphs are
based on local search methods \cite{berman2000d, berman2003optimizing, chandra2001greedy, halldorsson1995approximating, hurkens1989size}. 
The first such result by Hurkens and Schrijver \cite{hurkens1989size} gave a $(\frac{d}{2}+\epsilon)$-approximation algorithm. Halldorsson \cite{halldorsson1995approximating}
presented a quasi-polynomial $(\frac{d+2}{3})$-approximation algorithm for the
unweighted $d$-UHM. Sviridenko and Ward \cite{sviridenko2013large} established a
polynomial time $(\frac{d+2}{3})$-approximation algorithm that is the first polynomial time improvement over the $(\frac{d}{2}+\epsilon)$ result from \cite{hurkens1989size}. Cygan \cite{cygan2013improved} and Furer and Yu \cite{furer2013approximate} both provide a $(\frac{d+1}{3} + \epsilon)$ polynomial time
approximation algorithm, which is the best approximation guarantee known so far. On the other hand, Hazan, Safra and
Schwartz \cite{hazan2006complexity} proved that it is hard to approximate $d$-UHM problem
within a factor of $\Omega(d/\log{d})$.\\

\noindent\textbf{Matching in parallel computation models.} 
The study of the graph maximum matching problem in parallel computation models can be traced back to PRAM algorithms
of 1980s \cite{israeli1986fast,alon1986fast, luby1986simple}. Since then it has been studied in the LOCAL and MPC models, and $(1 + \epsilon)$- approximation can be achieved in 
$O(\log{\log{n}})$ rounds using a space of $O(n)$ \cite{assadi2019coresets, czumaj2019round, ghaffari2018improved}. The question of finding a maximal matching in a small number of rounds has also been considered in \cite{ghaffari2019sparsifying,lattanzi2011filtering}. Recently, Behnezhad \textit{et al}. \cite{behnezhad2019exponentially} presented a $O(\log{\log{\Delta}})$ round algorithm for maximal matching with $O(n)$ memory per machine. While we are not aware of previous work on hypergraph matching in
the MPC model, finding a maximal hypergraph matching has been considered in the LOCAL model.
Both Fischer \textit{et al}. \cite{fischer2017deterministic} and Harris \cite{harris2018distributed} provide a deterministic distributed algorithm that computes $O(d)$-approximation to the $d$-UHM.\\

\noindent\textbf{Maximum Independent Set.} The $d$-UHM is strongly related to different variants of the Independent Set problem as well as other combinatorial problems. The maximum independent set (MIS) problem on degree bounded graphs can be
mapped to $d$-UHM when the degree bound is $d$ \cite{halldorsson2004approximations, berman1994approximating,halldorsson1997greed, trevisan2001non}. The $d$-UHM problem can also be studied under a more general problem of
maximum independent set on $(d + 1)$-claw-free graphs \cite{berman2000d, chandra2001greedy}. (See \cite{chan2012linear} of connections between $d$-UHM and other combinatorial optimization problems).

\section{A 3-round $O(d^2)$-approximation}
\vspace{-2mm}
In this section, we generalize the randomized composable coreset algorithm of Assadi and Khanna \cite{assadi2017randomized}. They used a maximum matching as a coreset and obtained a $O(1)$-approximation. We use a hypergraph maximum matching and we obtain a $O(d^2)$-approximation. We first  define a  $k$-partitioning 
and then present our  greedy approach.

\begin{definition}[Random $k$-partitioning]
Let $E$ be an edge-set of a hypergraph $G(V, E)$. We say that a collection of edges
$E^{(1)}, \ldots , E^{(k)}$ is a
random $k$-partition of $E$ if the sets are constructed by assigning each edge $e \in E$ to some $E^{(i)}$
chosen uniformly at random. A random $k$-partition of $E$ naturally results in partitioning the graph
G into $k$ subgraphs $G^{(1)}, \ldots , G^{(k)}$ where $G^{(i)} := G(V, E^{(i)})$ for all $i \in [k]$.
\end{definition}



Let $G(V, E)$ be any $d$-uniform hypergraph and $G^{(1)}, \ldots , G^{(k)}$ be a random $k$-partitioning of $G$. We describe a simple greedy process for combining the maximum matchings of $G^{(i)}$, and prove that this process results in a $O(d^2)$-approximation of the maximum
matching of $G$.

\begin{algorithm}
\SetAlgoLined
\DontPrintSemicolon
 Construct a random $k$-partitioning of $G$
across the $k$ machines. Let $M^{(0)}:=\emptyset.$\;
 \For{$i = 1$ to $k$}{
  Set $MM(G^{(i)})$ to be an arbitrary
hypergraph maximum matching of $G^{(i)}$.\;
Let $M^{(i)}$ be a maximal matching obtained by adding to $M^{(i-1)}$ the edges of $MM(G^{(i})$ that do not violate the matching property.\;
}
\Return $M:=M^{(k)}$.\;
 \caption{\textsf{Greedy}}
\end{algorithm}

\begin{restatable}{theorem}{result1}
\label{result1}
Greedy computes, with high probability,
 a $\big(3d(d-1)+3+o(1)\big)$-approximation for the $d$-UHM problem in 3 $MPC$ rounds on machines of memory $s = \Tilde{O}(\sqrt{nm})$.
\end{restatable}
 
 In the rest of the section, we present the proof of Theorem \ref{result1}. Let $c =  \frac{1}{3d(d-1)+3}$, we show that $M^{(k)} \geq c \cdot \mu(G)$ w.h.p, where $M^{(k)}$ is the output of \textsf{Greedy}. The randomness stems from the fact that the matchings $M^{(i)}$ (for $i \in \{1,\ldots,k\}$) constructed by \textsf{Greedy} are random variables
depending on the random $k$-partitioning. We adapt the general approach in \cite{assadi2017randomized} for $d$-uniform hypergraphs, with $d \geq 3$. Suppose at the beginning of the
$i$-th step of \textsf{Greedy}, the matching $M^{(i-1)}$ is of size $o(\mu(G)/d^2)$. One can see that in this
case, there is a matching of size $\Omega(\mu(G))$ in $G$ that is entirely incident on vertices of $G$ that are
not matched by $M^{(i-1)}$. We can show that in fact $\Omega(\mu(G)/(d^2 k))$ edges of this matching are
appearing in $G^{(i)}$, even when we condition on the assignment of the edges in the first $(i-1)$ graphs.
Next we argue that the existence of these edges forces any maximum matching of
$G^{(i)}$
to match $\Omega(\mu(G)/(d^2 k))$  edges in $G^{(i)}$ between the vertices that are not matched by $M^{(i-1)}$.
These edges can always be added to the matching $M^{(i-1)}$ to form $M^{(i)}$. Therefore, while the
maximal matching in \textsf{Greedy} is of size $o(\mu(G))$, we can increase its size by $\Omega(\mu(G)/(d^2 k))$ 
edges in each of the first $k/3$ steps, hence obtaining a matching of size $\Omega(\mu(G)/d^2)$  at the end. The
following Lemma \ref{lone} formalizes this argument.

\begin{restatable}{lemma}{lone}
\label{lone}
For any $i \in [k/3]$, if $M^{(i-1)} \leq c \cdot \mu(G)$, then, w.p. $1 - O(1/n)$,
$$ M^{(i)} \geq M^{(i-1)}+ \frac{ 1-3d(d-1)c-o(1)}{k} \cdot \mu(G) \ .$$
\end{restatable}

Before proving the lemma, we define some notation. Let $M^*$ be an arbitrary maximum matching of $G$. For
any $i \in [k]$, we define $M^{*<i}$ as the part of $M^*$ assigned to the first $i - 1$ graphs in the random
$k$-partitioning, i.e., the graphs $G^{(1)}, \ldots G^{(i-1)}$. We have the following concentration result:

\begin{claim}\label{c1}
W.p. $1 - O(1/n)$, for any $i \in [k]$:
$$ M^{*<i} \leq \Big( \frac{i-1+o(i)}{k} \Big) \cdot \mu(G) $$
\end{claim}
\begin{proof}[Proof of claim \ref{c1}]
Fix an $i \in [k]$; each edge in $M^*$ is assigned to $G^{(1)}, \ldots , G^{(i-1)}$, w.p. $(i - 1)/k$, hence in
expectation, size of $M^{*<i}$ is $\frac{i-1}{k}\cdot \mu(G)$. For a large $n$, the ratio $\frac{\mu(G)}{k}$ is large and the claim follows from a standard application of
Chernoff bound.
\end{proof}

\begin{proof}[Proof of Lemma \ref{lone}]
Fix an $i \in [k/3]$ and the set of edges for $E^{(1)}, \ldots E^{(i-1)}$. This also fixes the matching $M^{(i-1)}$ while the set of edges in $E^{(i)}, \ldots , E^{(k)}$ together with the matching $M^{(i)}$ are still random variables. We further condition on the event that after fixing the edges in $E^{(1)}, \ldots , E^{(i-1)}, |M^{*<i}| \leq \frac{i-1+o(i)}{k} \cdot \mu(G)$ which happens w.p. $1-O(1/n)$ by claim \ref{c1}.

Let $V_{old}$ be the set of vertices incident on $M^{(i-1)}$ and $V_{new}$ be
the remaining vertices. Let $E^{\geq i}$ be the set of edges in $E \setminus E^{(1)} \cup \ldots \cup E^{(i-1)}$. We partition $E^{\geq i}$ into two parts: (i) $E_{old}$: the set of edges with \textit{at least one} endpoint in $V_{old}$, and (ii) $E_{new}$: the set
of edges incident entirely on $V_{new}$. Our goal is to show that w.h.p. any maximum matching of $G^{(i)}$
matches $\Omega(\mu(G)/k)$ vertices in $V_{new}$ to each other by using the edges in $E_{new}$. The lemma then
follows easily from this.

The edges in the graph $G^{(i)}$ are chosen by independently assigning each edge in
$E^{\geq i}$ to $G^{(i)}$ w.p. $1/(k - i + 1)$. This independence makes it possible to treat the edges in $E_{old}$ and $E_{new}$
separately. We can fix the set of sampled edges of $G^{(i)}$
in $E_{old}$, denoted by $E^{i}_{old}$, without changing
the distribution of edges in $G^{(i)}$ chosen from $E_{new}$. Let $\mu_{old} := MM(G(V, E^{i}_{old}))$, i.e., the maximum
number of edges that can be matched in $G^{(i)}$ using only the edges in $E^{i}_{old}$. In the following, we show
that w.h.p., there exists a matching of size $\mu_{old}+\Omega(\mu(G)/k)$ in $G^{(i)}$. By the definition of $\mu_{old}$, this implies that any maximum matching of $G^{(i)}$ has to use at least $\Omega(\mu(G)/k)$ edges in $E_{new}$.

Let $M_{old}$ be any arbitrary maximum matching of size $\mu_{old}$ in $G(V, E^{i}_{old})$. Let $V_{new}(M_{old})$ be
the set of vertices in $V_{new}$ that are incident on $M_{old}$. We show that there is a large matching in
$G(V, E_{new})$ that avoids $V_{new}(M_{old})$.\\

\begin{claim}\label{c12}
$| V_{new}(M_{old})| < c \cdot d(d-1)\cdot \mu(G)$.
\end{claim}
\begin{proof}[Proof of Claim \ref{c12}]
  Since any edge in $M_{old}$ has at least one endpoint in
$V_{old}$, we have $|V_{new}(M_{old})| \leq (d-1)|M_{old}| \leq (d-1) |V_{old}|$. By the assertion of the lemma, $|M^{(i-1)}| < c \cdot \mu(G)$, and hence $V_{new}(M_{old})| \leq (d-1)\cdot |V_{old}| \leq d(d-1)\cdot |M^{(i-1)}| < c \cdot d(d-1) \cdot \mu(G)$.
\end{proof}

\begin{claim}\label{c2}
There exists a matching in $G(V,E_{new})$ of size $\Big(\frac{k-i+1-o(i)}{k} - 2d(d-1)c\Big)\cdot \mu(G)$ that avoids the vertices of $V_{new}(M_{old})$.
\end{claim}
\begin{proof}[Proof of Claim \ref{c2}]
By the assumption that $|M^{*<i}| \leq \frac{i-1+o(i)}{k} \cdot \mu(G)$, there is a matching of size $\tfrac{k-i+1-o(i)}{k}\cdot \mu(G)$ in the graph $G(V,E^{\geq i})$.  By removing the edges in $M$ that are either incident on
$V_{old}$ or $V_{new}(M_{old})$, at most $2d(d-1)c \cdot \mu(G)$ edges are removed from $M$. Now the remaining matching
is entirely contained in $E_{new}$ and also avoids $V_{new(Mold)}$, hence proving the claim.
\end{proof}

We are now ready to finalize the proof of Lemma \ref{lone}. Let $M_{new}$ be the matching guaranteed by Claim \ref{c2}.
Each edge in this matching is chosen in $G^{(i)}$ w.p. $1/(k - i + 1)$ independent of the other edges.
Hence, by Chernoff bound, there is a matching of
size
$$ (1-o(1))\cdot \Big(\frac{1}{k} - \frac{o(i)}{k(k-i+1)} - \frac{2d(d-1)c}{k-i+1}\Big) \cdot \mu(G) \geq \frac{1 - o(1) - 3d(d-1)c}{k}\cdot \mu(G)$$
in the edges of $M_{new}$ that appear in $G^{(i)}$ (for $i \leq k/3$). This matching can be directly added to the matching
$M_{old}$, implying the existence of a matching of size $\mu_{old} +\frac{1 - o(1) - 3d(d-1)c}{k} \cdot \mu(G)$
 in $G^{(i)}$. As we argued before, this ensures that any maximum matching of $G^{(i)}$ contains at least 
$\frac{1 - o(1) - 3d(d-1)c}{k} \cdot \mu(G)$ edges
in $E_{new}$. These edges can always be added to $M^{(i-1)}$ to form $M^{(i)}$, hence proving the lemma.
\end{proof}

\begin{proof}[Proof of Theorem \ref{result1}]
Recall that $M := M^{(k)}$ is the output matching of \textsf{Greedy}. For the first
$k/3$ steps of \textsf{Greedy}, if at any step we got a matching of size at least $c \cdot \mu(G)$, then we are
already done. Otherwise, at each step, by Lemma \ref{lone}, w.p. $1 - O(1/n)$, we increase the size of
the maximal matching by $\frac{1  - 3d(d-1)c - o(1)}{k}\cdot \mu(G)$ edges; consequently, by taking a union bound on
the $k/3$ steps, w.p. $1 - o(1)$, the size of the maximal matching would be $\frac{1  - 3d(d-1)c - o(1)}{3}\cdot \mu(G)$. Since $c = 1/(3d(d-1)+3)$, we ensure that $\frac{1  - 3d(d-1)c}{3} = c $ and in either case, the matching computed by \textsf{Greedy} is of size at
least $\mu(G)/(3d(d-1) + 3) - o(\mu(G))$, and this proves that \textsf{Greedy} is a $O(d^2)$-approximation. All is left is to prove that \textsf{Greedy} can be implemented in 3 rounds with a memory of $\Tilde{O}(\sqrt{nm})$ per machine.  Let $k = \sqrt{\frac{m}{n}}$ be the number of machines, each with a memory of $\Tilde{O}(\sqrt{nm})$. We claim that \textsf{Greedy} can run in three rounds.  In the first round, each machine randomly partitions the
edges assigned to it across the $k$ machines. This results in a random $k$-partitioning of the graph
across the machines. In the second round, each machine sends a maximum matching
of its input to a designated central machine $M$; as there are $k = \sqrt{\frac{m}{n}}$ machines and each machine
is sending $\tilde{O}(n)$ size coreset, the input received by $M$ is of size
$\Tilde{O}(\sqrt{nm})$ and hence can be stored
entirely on that machine. Finally, $M$ computes the answer by combining the matchings. \end{proof}

We conclude this section by stating that 
computing a maximum matching on every machine is only required for the analysis, i.e., to show that
there exists a large matching in the union of coresets. In practice, we can use an approximation algorithm
to obtain a large matching from the coresets. In our experiments, we use a maximal matching instead of maximum.
\section{A $O(\log{n})$-rounds  $d$-approximation algorithm}

In this section, we show
how to compute a maximal matching in $O(\log{n})$ MPC rounds, by generalizing the algorithm  in \cite{lattanzi2011filtering} to $d$-uniform hypergraphs. The algorithm provided by Lattanzi \textit{el al.} \cite{lattanzi2011filtering} 
 computes a maximal matching in graphs in $O(\log n)$ if $s = \Theta(n)$, and in at most$\lfloor c/\epsilon\rfloor$ iterations when $s = \Theta(n^{1+\epsilon})$, where $0< \epsilon < c$ is a fixed constant. Harvey \textit{et al.} \cite{harvey2018greedy} show a similar result.

The algorithm first samples $O(s)$ edges and finds
a maximal matching $M_1$ on the resulting subgraph (we will specify a bound on the memory $s$ later). Given this
matching, we can safely remove edges that are in conflict (i.e.
those incident on nodes in $M_1$) from the original graph $G$. If the
resulting filtered graph  $H$ is small enough to fit onto a single machine,
the algorithm augments $M_1$ with a matching found on $H$.
Otherwise, we augment $M_1$
with the matching found by recursing
on $H$. Note that since the size of the graph reduces from round to
round, the effective sampling probability increases, resulting in a
larger sample of the remaining graph.

\begin{restatable}{theorem}{mpcmax}
\label{mpcmax}
Given a $d$-uniform hypergraph $G(V, E)$, \textsf{Iterated-Sampling} omputes a maximal matching in $G$ with high probability in $O(\log{n})$ $MPC$ rounds on machines of memory $s =\Theta(d\cdot n)$.
\end{restatable}

\begin{algorithm}
\SetAlgoLined
\DontPrintSemicolon
 Set $M:=\emptyset$ and $\mathcal{S}= E$.\;
 Sample every edge $e \in \mathcal{S}$ uniformly with probability $p =\frac{s}{5|\mathcal{S}|d}$ to form $E'$.\;
 If $|E'| > s$ the algorithm fails. Otherwise give the graph $G(V,E')$ as input to a single machine and compute a maximal matching $M'$ on it. Set $M= M\cup M'$.\;
 Let $I$ be the  unmatched vertices in $G$ and  $G[I]$ the induced subgraph with edges  $E[I]$. If $E[I] > s$, set $\mathcal{S}:= E[I]$ and return to step 2.\;
 Compute a maximal matching $M''$ on $G[I]$ and output $M = M \cup M''$.\;
 
 \caption{\textsf{Iterated-Sampling}}
	 \label{maximalalg}
\end{algorithm}

Next we present the proof of Theorem \ref{mpcmax}. We first show that after sampling edges, the size of the graph 
$G[I]$ induced by unmatched vertices decreases exponentially with high probability, and therefore the algorithm terminates in $O(\log{n})$ rounds. Next we argue that $M$ is indeed a maximal matching by showing that if after terminating, there is an edge that's still unmatched, this will yield a contradiction. This is formalized in the following lemmas.

\begin{lemma}\label{smpl}
 Let $E' \subset E$ be a set of edges chosen independently
with probability $p$. Then with probability at least $1 - e^{-n}$,
for all $I \subset V$ either $|E[I]| < 2n/p$ or $E[I] \cap E' \neq \emptyset$.
 \end{lemma}
\begin{proof}[Proof of Lemma \ref{smpl}]
 Fix one such subgraph, $G[I]=(I,E[I])$ with
$|E[I]| \geq 2n/p$. The probability that none of the edges in $E[I]$
were chosen to be in $E'$ is $(1 - p)^{|E[I]|} \leq (1 - p)^{2n/p} \leq e^{-2n}$. Since there are at most $2^n$ total possible induced subgraphs $G[I]$ (because each vertex is either matched or unmatched),
the probability that there exists one that does not have an edge in
$E'$ is at most $2^n e^{-2n} \leq e^{-n}$.
\end{proof}

\begin{lemma}\label{cmplx}
If $s \geq 20d \cdot n$ then \textsf{Iterated-Sampling} runs for at most $O(\log{n})$ iterations with high probability.

\end{lemma}
\begin{proof}[Proof of Lemma \ref{cmplx}]
 Fix an iteration $i$ of \textsf{Iterated-Sampling} and let $p$ be the sampling
probability for this iteration. Let $E_i$ be the set of edges at the
beginning of this iteration, and denote by $I$ be the set of unmatched
vertices after this iteration. From Lemma \ref{smpl}, if $|E[I]| \geq 2n/p$
then an edge of $E[I]$ will be sampled with high probability. Note
that no edge in $E[I]$ is incident on any edge in $M'$. Thus, if an
edge from $E[I]$ is sampled then \textsf{Iterated-Sampling} would have chosen
this edge to be in the matching. This contradicts the fact that no
vertex in $I$ is matched. Hence, $|E[I]| \leq 2n/p$ with
high probability.

Now consider the first iteration of the algorithm, let $G_1(V_1, E_1)$
be the induced graph on the unmatched nodes after the first step of
the algorithm. The above argument implies that $|E_1| \leq \frac{10d\cdot n |E_0|}{s} \leq \frac{10d\cdot n |E|}{s} \leq \frac{|E|}{2}$. Similarly $|E_2| \leq \frac{10d\cdot n |E_1|}{s} \leq \frac{(10d\cdot n)^2 |E_1|}{s^2} \leq \frac{|E|}{2^2}$. After $i$ iterations : $|E_i| \leq \frac{|E|}{2^i}$, and the algorithm will terminate after $O(\log{n})$ iterations.
\end{proof}

\begin{lemma}\label{correctness}
\textsf{Iterated-Sampling} finds a maximal matching of $G$ with high probability.
\end{lemma}
\begin{proof}[Proof of Lemma \ref{correctness}]
 First consider the case that the algorithm does not fail. Suppose there is an edge $e=\{v_1,\ldots,v_d\} \in E$ such that none of the $v_i$'s are matched in the final matching $M$ that the algorithm output. In the last iteration of the algorithm, since $e\in E$ and its endpoints are not matched, then $e \in E[I]$. Since this is the last run of the algorithm, a maximal matching $M''$ of $G[I]$ is computed on one machine. Since $M''$ is maximal, at least one of the $v_i$'s must be matched in it. All of the edges of $M''$ get added to $M$ in the last step. This yields a contradiction.
 
 Next, consider the case that the algorithm fails. This occurs
due to the set of edges $E'$ having size larger than the memory in some iteration
of the algorithm.  Note that $\mathbb{E}[|E'|] = |S| \cdot p = s/5d \leq s/10$ in a
given iteration. By the Chernoff Bound it follows that $|E'| \geq s$ with probability smaller than $2^{s} \geq 2^{-20d\cdot n}$ (since $s \geq 20d\cdot n$). By
Lemma \ref{cmplx} the algorithm completes in at most $O(\log{n})$ rounds,
thus the total failure probability is bounded by $O(\log{n} \cdot 2^{-20\cdot n})$ using
the union bound.
\end{proof}
We are now ready to show that \textsf{Iterated-Sampling} can be implemented in $O(\log n)$ MPC rounds.
\begin{proof}[Proof of Theorem \ref{mpcmax}]
We show that \textsf{Iterated-Sampling} can be implemented in the $MPC$ model with
machines of memory $\Theta(dn)$ and $O(\log{n})$ MPC rounds. Combining this with Lemma \ref{correctness} on the
correctness of \textsf{Iterated-Sampling} , we immediately obtain Theorem \ref{mpcmax} for the case of $s = \Theta(dn)$. Every iteration of \textsf{Iterated-Sampling} can be implemented in $O(1)$ MPC rounds. Suppose the edges are initially distributed over all machines. We can sample every edge with the probability $p$ (in each machine) and then send the sampled edges $E'$ to the first machine. With high probability we will have $|E'| = O(n)$, so we can fit the sampled edges in the first machine. Therefore, the sampling can be done in one $MPC$ round. Computing the subgraph of $G$ induced by $I$ can be done in 2 rounds; one round to send the list of unmatched vertices $I$ from the first machine to all the other machines, and the other round to compute the subgraph $G[I]$, and send it to a single machine if it fits, or start the sampling again.
\end{proof}
\section{A 3-round $O(d^3)$-approximation using HEDCSs}
In graphs,  edge degree constrained subgraphs (EDCS) \cite{bernstein2015fully,bernstein2016faster} have been used as a local condition for identifying large matchings, and leading to good dynamic and parallel algorithms \cite{bernstein2015fully, bernstein2016faster, assadi2019coresets}. These papers showed that an EDCS exists, that it contains a large matching, that it is robust to sampling and composition, and that it can be used as a coreset.

In this section, we present a generalization of EDCS for hypergraphs. We prove that an HEDCS exists, that it contains a large matching, that it is robust to sampling and composition, and that it can be used as a coreset.  The proofs and algorithms, however, require significant developments beyond the graph case.
We first present definitions of a fractional matching and HEDCS.

\begin{definition}[Fractional matching]
In a hypergraph $G(V,E)$, a fractional matching is a mapping from $y: E \mapsto [0,1]$ such that for every edge $e$ we have $0 \leq y_e \leq 1$ and for every vertex $v$ : $\sum_{e \in v} y_e \leq 1$.
The value of such fractional matching is equal to $\sum_{e \in E} y_e$. For $\epsilon > 0$, a fractional matching is an $\epsilon$-restricted fractional matching if for every edge $e$: $y_e = 1 \mbox{ or } y_e \in [0,\epsilon].$
\end{definition}


\begin{definition}
For any hypergraph $G(V, E)$ and integers $\beta \geq \beta^- \geq 0$ a hyperedge degree constraint
subgraph HEDCS$(H, \beta, \beta^-)$ is a subgraph $H := (V, E_{H})$ of $G$ satisfying:
\begin{itemize}
    \item \textit{(P1)}: For any hyperedge $e \in E_{H}$: $ \sum_{v \in e} d_{H}(v) \leq \beta.$
   \item \textit{(P2)}:  For any hyperedge $e \not\in E_{H}$: $ \sum_{v \in e} d_{H}(v) \geq \beta^-.$
\end{itemize}
\end{definition}

We show via a constructive proof that a hypergraph contains an HEDCS when the parameters of this HEDCS satisfy the inequality $\beta - \beta^- \geq d - 1$. 
\begin{restatable}{lemma}{existence}
\label{existence}
Any hypergraph $G$ contains an HEDCS$(G, \beta, \beta^-)$ for any parameters $\beta - \beta^- \geq d - 1$.
\end{restatable}

\begin{proof}
 Consider the following simple procedure for creating an HEDCS $H$ of a given hypergraph $G$: start by initializing $H$ to be equal to $G$. And while $H$ is not an HEDCS($G, \beta, \beta^-)$, find an edge $e$ which violates one of the properties of HEDCS and fix it. Fixing the edge $e$ implies removing it from $H$ if it was violating Property \textit{(P1)} and adding it to
$H$ if it was violating Property \textit{(P2)}. The output of the above procedure is clearly an HEDCS of graph $G$. However, a-priori it is not
clear that this procedure ever terminates as fixing one edge $e$ can result in many edges violating the
HEDCS properties, potentially undoing the previous changes. In the following, we use a potential
function argument to show that this procedure always terminates after a finite number of steps,
hence implying that a $HEDCS(G, \beta, \beta^{-})$ always exists.\\

We define the following potential function $\Phi$:
$$ \Phi := (\frac{2}{d}\beta - \frac{d-1}{d}) \cdot \sum\limits_{v \in V} d_{H}(v) - \sum\limits_{e \in H} \sum\limits_{u \in e} d_{H}(u)$$

We argue that in any step of the procedure above, the value of $\Phi$ increases by at least 1. Since the
maximum value of $\Phi$ is at most $O(\frac{2}{d}n\cdot \beta^2)$, this immediately implies that this procedure terminates
in $O(\frac{2}{d}n\cdot \beta^2)$ iterations.\\

Define $\Phi_1 =  (\frac{2}{d}\beta - \frac{d-1}{d}) \cdot \sum\limits_{v \in V} d_{H}(v)$ and $\Phi_2 = \sum\limits_{e \in H} \sum\limits_{u \in e} d_{H}(u)$. Let $e$ be the edge we choose to fix at this step, $H_b$ be the subgraph
before fixing the edge $e$, and $H_a$ be the resulting subgraph. Suppose first that the edge $e$ was violating Property \textit{(P1)} of HEDCS. As the only change is in 
the degrees of vertices $v \in e$, $\Phi_1$ decreases by $(2\beta - (d-1))$. On the other hand, $ \sum\limits_{v \in e} d_{H_{b}}(v) \geq \beta + 1$
originally (as $e$ was violating Property (P1) of HEDCS), and hence after removing $e$, $\Phi_2$
increases by $\beta + 1$. Additionally, for each edge $e_u$ incident upon $u \in e$ in $H_a$
, after removing the edge
$e$, $\sum\limits_{v \in e_u}d_{H_a}(v)$ decreases by one.  As there are at least $\sum\limits_{u \in e}d_{H_a}(u) = \sum\limits_{u \in e}d_{H_b}(u) - d  \geq \beta - (d-1)$ 
choices for $e_u$, this means that in total, $\Phi_2$ increases by at least $2\beta + 1 - (d-1)$. As a result,
in this case $\Phi$ increases by at least 1 after fixing the edge $e$.

Now suppose that the edge $e$ was violating Property \textit{(P2)} of HEDCS instead. In this case,
degree of vertices $u \in e$ all increase by one, hence $\Phi_1$ increases by $2\beta - (d-1)$. Additionally,
note that since edge $e$ was violating Property (P2) we have $ \sum\limits_{v \in e} d_{H_b}(v) \leq \beta^- - 1$, so the addition of edge $e$ decreases $\Phi_2$ by at most $\sum\limits_{v \in e} d_{H_a}(v) = \sum\limits_{v \in e} d_{H_b}(v) + d \leq \beta^- - 1 + d $. Moreover, for each edge $e_u$ incident upon $u \in e$, after adding the edge $e$, $\sum\limits_{v \in e_u} d_{H_a}(v)$ increases by one and
since there are at most $ \sum\limits_{v \in e} d_{H_b}(v) \leq \beta^- - 1$ choices for $e_u$, $\Phi_2$ decreases in total by at most $2\beta^- - 2 +d$. The total variation in $\Phi$ is therefore equal to $2\beta - (d-1) - (2\beta^- - 2 +d) = 3 + 2(\beta - \beta^- - d).$ So if $\beta - \beta^- \geq d - 1$, we have that $\Phi$ increases by at least 1 after fixing edge $e$.
\end{proof}

The main result in this section shows that an HEDCS of a graph contains a large $\epsilon$-restricted fractional matching that approximates the maximum matching by a factor less than $d$.

\begin{restatable}{theorem}{mainth}
\label{mainth}
Let $G$ be a $d$-uniform hypergraph and $0 \leq \epsilon < 1$. Let $H$:= HEDCS$(G, \beta, \beta(1 - \lambda))$ with $\lambda =\frac{\epsilon}{6}$ and $\beta \geq \frac{8 d^2}{d-1} \cdot \lambda^{-3}$. Then $H$ contains an $\epsilon$-restricted fractional matching $M^{H}_f$ with total value at least $\mu(G)(\frac{d}{d^2 - d + 1 } - \frac{\epsilon}{d-1}).$
\end{restatable}

In order to prove Theorem \ref{mainth}, we will need the following two lemmas. The first lemma we prove is an algebraic result that will help us bound the contribution of vertices in the $\epsilon$-restricted fractional matching. The second lemma identifies  additional structure on the HEDCS, that we will use to construct the fractional matching of the theorem. For proofs of both lemmas, see Appendix~\ref{appmainth}

\begin{restatable}{lemma}{phix}
\label{phi}
Let $ \phi(x) = \min \{ 1 , \frac{(d-1)x}{d (\beta-x)} \} $. If $a_1, \ldots a_d \geq 0$ and $a_1 + \ldots + a_d \geq \beta(1 - \lambda)$ for
some $\lambda \geq 0$, then $\sum\limits_{i=1}^d \phi(a_i) \geq 1 - 5\cdot \lambda$.
\end{restatable}

\begin{restatable}{lemma}{fourprop}
\label{4prop}
Given any $HEDCS(G,\beta, \beta(1-\lambda))$ $H$, we can find two disjoint sets of vertices $X$ and $Y$ that satisfy the following properties:
\begin{enumerate}
    \item $|X| + |Y| = d \cdot \mu(G)$.
    \item There is a perfect matching in $Y$ using edges in $H$.
    \item Letting $\sigma = \frac{|Y|}{d} + \sum\limits_{x \in X} \phi(d_H(x)) $,we have that $\sigma \geq \mu(G) (1-5\lambda)$.
    \item All edges in $H$ with vertices in $X$ have at least one other vertex in $Y$, and have vertices only in $X$ and $Y$.
\end{enumerate}
\end{restatable}

We are now ready to prove Theorem \ref{mainth}.

\begin{proof}[Proof of theorem \ref{mainth}]
Suppose we have two sets $X$ and $Y$ satisfying the properties of Lemma \ref{4prop}. We construct an $\epsilon$-restricted fractional matching $M^H_f$
using the edges in $H$ such that \[val(M^H_f) \geq (\frac{d}{d^2 - d + 1  } - \frac{\epsilon}{d-1}) \mu(G),\] where $val(M^H_f)$ is the value of the fractional matching $M^H_f$. Now, by Property 2 of Lemma \ref{4prop},
$|Y|$ contains a perfect matching $M^H_Y$ using edges in
$H$. Let $Y^-$ be a subset of $Y$ obtained by randomly sampling exactly $1/d$ edges of $M^H_Y$ and adding their endpoints to $Y^-$ Let $Y^* = Y \setminus Y^-$ and observe that $|Y^-| = |Y|/d$ and $  |Y^*| = \frac{d-1}{d} |Y|$.

Let $H^*$ be the subgraph of $H$ induced
by $X \cup Y^*$ (each edge in $H^*$ has vertices in only $X$ and $Y^*$). We define a fractional matching $M^{H^*}_f$
on the edges of $H^*$ in which all edges have value at most
$\epsilon$. We will then let our final fractional matching $M^H_f$
be the fractional matching  $M^{H^*}_f$
joined with the perfect
matching in $H$ of $Y^-$ (so $M^H_f$ assigns value $1$ to the
edges in this perfect matching). $M^H_f$ is, by definition, an
$\epsilon$-restricted fractional matching.

We now give the details for the construction of $M^{H^*}_f$. Let $V^* = X \cup Y^*$ be the vertices of $H^*$, and let
$E^*$ be its edges. For any vertex $v \in V^*$, define $d^*_{H}(v)$ to be the degree of $v$ in $H^*$. Recall that by Property 4 of
Lemma \ref{4prop}, if $x \in X$ then all the edges of $H$ incident to
$x$ go to $Y$ (but some might go to $Y^-$). Thus, for $x \in X$,
we have $E[d^*_{H}(x)] \geq \frac{d_{H}(x)(d-1)}{d}$.

We now define $M^{H^*}_f$ as follows. For every $x \in X$,
we arbitrarily order the edges of $H$ incident to $x$, and
then we assign/add a value of $\min\Big\{\frac{\epsilon}{|X \cap e|}, \frac{1}{|X \cap e|}\frac{1}{\beta -d_H(x)}\Big\}$
to these edges
one by one, stopping when either $val(x)$ (the sum of values assigned to vertices incident to $x$) reaches 1 or
there are no more edges in $H$ incident to $x$, whichever
comes first. In the case that $val(x)$ reaches $1$ the last
edge might have added value less than $\min\Big\{\frac{\epsilon}{|X \cap e|}, \frac{1}{|X \cap e|} \frac{1}{\beta -d_H(x)}\Big\}$, where $e$ is the last edge to be considered.

We now verify that $M^{H^*}_f$ is a valid fractional matching in that all vertices have value at most ~1. This is clearly
true of vertices $x \in X$ by construction. For a vertex
$y \in Y^*$, it suffices to show that each edge incident to
$y$ receives a value of at most $1/d_H(y) \leq 1/d^*_H(y) $.  To
see this, first note that the only edges to which
$M^{H^*}_f$ assigns non-zero values have at least two endpoints in $X \times Y^*$. Any
such edge $e$ receives value at most $\min{\big\{ \epsilon, \sum\limits_{x \in X \cap e} \frac{1}{|X\cap e|}\frac{1}{\beta - d_H(x)}\big\} }$, but since $e$ is in $M^{H^*}_f$
and so in $H$, we have by
Property \textit{(P1)} of an HEDCS that $d_H(y) \leq \beta - d_H(x)$, and
so $\sum\limits_{x \in X \cap e} \frac{1}{|X\cap e|} \frac{1}{\beta -d_H(x)} \leq  \frac{1}{|X\cap e|} \frac{|X \cap e|}{d_H(y)} \leq \frac{1}{d_H(y)}$ .

By construction, for any $x \in X$, we have that the value $val(x)$ of $x$ in $M_f^{H^*}$  satisfies :
\begin{eqnarray*}
  val(x) & = & \min \left\{ 1 , \sum\limits_{e \ni x}  \min\left\{ \frac{\epsilon}{|X\cap e|}, \frac{1}{|X \cap e|} \frac{1}{\beta - d_H(x)} \right\} \right\} \\
  & \geq &\min \left\{ 1 , d_H^*(x) \cdot \min\left\{ \frac{\epsilon}{d-1}, \frac{1}{d-1}\frac{1}{\beta - d_H(x)} \right\} \right\}\ . 
\end{eqnarray*}
Furthermore, we can bound the value of the fractional matching $M^{H*}_f$: as
$ val(M^{H*}_f) \geq \sum\limits_{x \in X} val(x).$
For convenience,  we use
$val'(x) = \min \left\{ 1 , d_H^*(x) \cdot \min\left\{ \epsilon, \frac{1}{\beta - d_H(x)} \right\} \right\}$
such that

\begin{eqnarray}
    val(x)  & \geq & \frac{val'(x)}{d-1} \, \mbox{ and}\\
     val(M^{H*}_f)  & \geq & \frac{1}{d-1}\sum\limits_{x \in X} val'(x) \ ,\  \label{val}
\end{eqnarray}


Next we present a lemma, proved in Appendix \ref{appmainth} that bounds $val'(x)$ for each vertex.
\begin{restatable}{lemma}{lemmaval}
\label{lemmaval}
For any $x \in X$, $E[val'(x)] \geq (1-\lambda)\phi(d_H(x))$.
\end{restatable}

This last lemma, combined with (\ref{val}), allows us to lower bound the value of $M_f^{H^*}$ : 
                                              $$ val(M_f^{H^*}) \geq \frac{1}{d-1}\sum\limits_{x \in X} val'(x)\geq \frac{1-\lambda}{d-1} \sum\limits_{x \in X} \phi(d_H(x)). $$

Note that we have constructed $M_f^H$ by taking the fractional value in $M_f^{H^*}$ and adding the perfect matching on edges from $Y^-$. The latter matching has size $\frac{|Y^-|}{d} = \frac{|Y|}{d^2}$, and the value of $M_f^{H}$ is bounded by:

\begin{eqnarray*}
val(M^H_f) & \geq  & \frac{1}{d-1} (1-\lambda) \sum\limits_{x \in X} \phi(d_H(x)) + \frac{|Y|}{d^2}  \\ 
      &  = & \frac{1}{d-1} \big( (1-\lambda) \sum\limits_{x \in X} \phi(d_H(x)) + \frac{|Y|}{d}\big) - \frac{
     |Y|}{d^2(d-1)} \\
     & \geq  & \frac{1}{d-1}(1 - \lambda)(1 - 5\lambda)\mu(G)- \frac{|Y|}{d^2(d-1)} \\ 
     & \geq  &   \frac{1}{d-1}(1 - 6\lambda)\mu(G) - \frac{|Y|}{d^2(d-1)} \ .
\end{eqnarray*}

To complete the proof, recall that $Y$ contains a perfect
matching in $H$ of $|Y|/d$ edges, so if $\frac{|Y|}{d} \geq \frac{d}{d(d-1)+1} \mu(G)$ then
there already exists a matching in $H$ of size at least  $\frac{d}{d(d-1)+1} \mu(G)$, and the theorem
is true. We can thus assume that $|Y|/d <\big(\frac{d}{d(d-1)+1}\big) \mu(G)$,
in which case the previous equation yields that:
\begin{eqnarray*}
       val(M^H_f) & \geq   &  \frac{1}{d-1}(1 - 6\lambda)\mu(G) - \frac{|Y|}{d^2(d-1)}  \\
   & \geq &  (\frac{1-6 \lambda}{d-1} )\mu(G) - \frac{\mu(G)}{(d-1)(d(d-1)+1)} \\
   & =  & \big(\frac{d}{d(d-1)+1} - \frac{6\lambda}{d-1}\big) \mu(G) \ .
\end{eqnarray*}

In both cases we get that
$$ val(M^H_f) \geq (\frac{d}{d^2 - d + 1 } - \frac{6\lambda}{d-1})\mu(G).$$
\end{proof}

\subsection{Sampling and constructing an HEDCS in the MPC model}

Our results in this section are general and applicable to every computation model. We prove structural properties about the HEDCSs that will help us construct them in the MPC model. We show that the degree distributions of every HEDCS (for
the same parameters $\beta$ and $\lambda$) are almost identical. In other words, the degree of any vertex $v$ is
almost the same in every HEDCS of the same hypergraph $G$. We show also that HEDCS are robust under edge sampling, i.e. that edge sampling from and HEDCS yields another HEDCS, and that the degree distributions of any two HEDCS for two different edge sampled subgraphs of $G$ is almost the same no matter how
the two HEDCS are selected. In the following lemma, we argue that any two HEDCS of a graph $G$ (for the same parameters $\beta$, $\beta^-$)
are ``somehow identical" in that their degree distributions are close to each other. In the rest of this section, we fix the parameters $\beta$, $\beta^-$ and two subgraphs $A$ and $B$ that are both HEDCS$(G,\beta,\beta^-)$.

\begin{restatable}{lemma}{degreedist}
\label{degree_dist}
(Degree Distribution Lemma). Fix a $d$-uniform hypergraph $G(V, E)$ and parameters $\beta$, $\beta^- = (1-\lambda)\cdot \beta $
(for $\lambda$ small enough). For any two subgraphs $A$ and $B$ that are HEDCS$(G,\beta, \beta^-)$, and any vertex $v \in V$, then 
$ |d_A(v) - d_B(v)| = O(\sqrt{n})\lambda^{1/2} \beta$.
\end{restatable}

 \begin{proof}
Suppose that we have $d_A(v) = k \lambda \beta$ for some $k$ and that $d_B(v) = 0$. We will show that if the $k = \Omega(\tfrac{\sqrt{n}}{\lambda^{1/2}})$, then this will lead  to a contradiction. Let $e$ be one of the $k \lambda \beta$ edges that are incident to $v$ in $A$. $e \not\in B$ so $\sum\limits_{u \neq v} d_B(u) \geq (1-\lambda)\beta$. From these $(1-\lambda)\beta$ edges, at most $(1-k\lambda)\beta$ can be in $A$ in order to respect \textit{(P1)}, so at least we will have $(k-1)\lambda \beta$ edges in $B \setminus A$, thus we have now covered $k\lambda\beta + (k-1)\lambda \beta $ edges in both $A$ and $B$.
 Let's keep focusing on the edge $e$, and especially on one of its $(k-1)\lambda\beta$ incident edges in $B \setminus A$. Let $e_1$ be such an edge. $e_1 \in B \setminus A $, therefore $\sum\limits_{v' \in e_1} d_A(v') \geq (1-\lambda)\beta$. The edges incident to $e_1$ in $A$ that we have covered so far are at most $(1-k\lambda)\beta$, therefore we still need at least $(k-1)\lambda$ new edges in $A$ to respect \textit{(P1)}. Out of these $(k-1)\lambda$ edges, at most $\lambda \beta$ can be in $B$ (because $e_1$ has already $(1-\lambda)\beta$ covered edges incident to it in $B$). Therefore at least $(k-2)\lambda \beta$ are in $A\setminus B$. Thus, we have so far covered  at least $k\lambda \beta + (k-1) \lambda \beta + (k-2)\lambda \beta$. One can see that we can keep doing this until we cover at least $\tfrac{k(k+1)}{2}  \lambda \beta$ edges in both $A$ and $B$. The number of edges in each of $A$ and $B$ cannot exceed $n\cdot \beta$ (each vertex has degree $\leq \beta$), therefore we will get a contradiction if 
$\tfrac{k(k+1)}{2}  \lambda \beta > 2n \beta$, which holds if $k > \tfrac{2\sqrt{n}}{\lambda^{1/2}}$. Therefore \[k \leq \tfrac{2\sqrt{n}}{\lambda^{1/2}}\mbox{, and } d_A(v) = k \lambda \beta \leq 2 \sqrt{n} \lambda^{1/2} \beta.\]
\end{proof} 

The next corollary shows that if the hypergraph is linear (every two hyperedges intersect in at most on vertex), then the degree distribution is closer. The proof is in Appendix \ref{section53}.

\begin{restatable}{corollary}{degreedistlinear}
\label{degree_dist_linear}
For $d$-uniform linear hypergraphs, the degree distribution is tighter, and
$ |d_A(v) - d_B(v)| = O(\log{n})\lambda \beta.$
\end{restatable}

Next we prove two lemmas regarding the structure of different HEDCSs across sampled subgraphs. The first lemma shows that edge sampling an HEDCS results in another HEDCS for the
sampled subgraph. The second lemma shows that the degree distributions
of any two HEDCS for two different edge sampled subgraphs of $G$ is almost the same.
\begin{restatable}{lemma}{sampledhedcs}
\label{sampledhedcs}
Let $H$ be a HEDCS($G, \beta_H , \beta^-_H)$ for parameters $\beta_H := (1 - \frac{\lambda}{\alpha})\cdot \frac{\beta}{p}$, $\beta_H^- := \beta_H - (d-1)$ and $\beta \geq 15d(\alpha d)^2 \cdot \lambda^{-2} \cdot \log{n}$ such that $p < 1 - \frac{2}{\alpha}$. Suppose $G_p := G^E_p(V, E_p)$ is an edge sampled subgraph of G and $H_p := H \cap G_p$; then, with high probability: 
\begin{enumerate}
    \item For any vertex $v \in V$ : $|d_{H_p}(v)- p \cdot d_H(v)| \leq \frac{\lambda}{\alpha d}\beta$
    \item $H_p$ is a HEDCS of $G_p$ with parameters $(\beta,(1 - \lambda) \cdot \beta)$ .
\end{enumerate}
\end{restatable}

\begin{proof}
Let $\alpha':= \alpha d$. For any vertex $v \in V$ , $E[d_{H_p}(v)] = p \cdot d_H(v)$ and $d_H(v) \leq \beta_H$ by Property \textit{(P1)} of HEDCS
$H$. Moreover, since each edge incident upon $v$ in $H$ is sampled in $H_p$ independently, by the Chernoff bound:
$$ P\Big(|d_{H_p}(v)- p \cdot d_H(v)| \geq \frac{\lambda}{\alpha'}\beta\Big) \leq 2\cdot \exp(-\frac{\lambda^2 \beta}{3 \cdot \alpha'^2}) \leq \frac{2}{n^{5d}}\ .$$
In the following,
we condition on the event that: $$ |d_{H_p}(v)- p \cdot d_H(v)| \leq \frac{\lambda}{\alpha'}\beta\ .$$ 
This event happens with probability at least $1 - \frac{2}{n^{5d-1}}$ by above equation and a union bound on
$|V| = n$ vertices. This finalizes the proof of the first part of the claim. We are now ready to prove
that $H_p$ is indeed am HEDCS$(G_p, \beta,(1 - \lambda) \cdot \beta)$ conditioned on this event. Consider any edge $e \in H_p$. Since $H_p \subset H$, $e \in H$ as well. Hence, we have,
$$ \sum\limits_{v \in e} d_{H_p}(v) \leq p \cdot \beta_H + \frac{d \lambda}{\alpha'} \beta = (1 - \frac{\lambda}{\alpha} + \frac{d\lambda}{\alpha'})\beta = \beta \ ,$$ because $\frac{\alpha}{\alpha'} = \frac{1}{d}$, where the inequality is by Property \textit{(P1)} of HEDCS $H$ and the equality is by the choice of
$\beta_H$. As a result, $H_p$ satisfies Property \textit{(P1)} of HEDCS for parameter $\beta$. Now consider an edge $e \in G_p \setminus H_p$. Since $H_p = G_p \cap H$, $e \not\in H$ as well. Hence,

\begin{eqnarray*}
\sum\limits_{v \in e} d_{H_p}(v) \geq p \cdot \beta_H^- - \frac{d \lambda}{\alpha'} \beta & = &  (1 - \frac{\lambda}{\alpha} -  \frac{d \lambda}{\alpha'}) \beta -  p\cdot (d-1) \\
& = & (1 - \frac{2\lambda}{\alpha}) \beta -  p\cdot (d-1) \\
& > & (1-\lambda) \cdot \beta \ . 
\end{eqnarray*}

\end{proof}

\begin{restatable}{lemma}{edgesample}
\label{lemma25}
(HEDCS in Edge Sampled Subgraph). Fix any hypergraph $G(V, E)$ and $p \in (0, 1)$. Let $G_1$ and
$G_2$ be two edge sampled subgraphs of $G$ with probability $p$ (chosen not necessarily independently).
Let $H_1$ and $H_2$ be arbitrary HEDCSs of $G_1$ and $G_2$ with parameters $(\beta,(1 - \lambda) \cdot \beta)$. Suppose $\beta \geq 15d(\alpha d)^2 \cdot \lambda^{-2} \cdot \log{n}$, then, with probability $1 - \frac{4}{n^{5d-1}}$, simultaneously for all $v \in V$ :
$|d_{H_1}(v) - d_{H_2}(v)| = O\big(n^{1/2}\big) \lambda^{1/2} \beta  $.
\end{restatable}

\begin{proof}
Let $H$ be an HEDCS(G, $\beta_H , \beta_H^-)$
for the parameters
$\beta_H$ and $\beta_H^-$ as defined in the previous lemma. The existence of $H$ follows since $\beta_H - (d-1) \geq \beta_H^-$ . Define $\hat{H}_1 := H \cap G_1$
and $\hat{H}_2 := H \cap G_2$. By Lemma \ref{sampledhedcs}, $\hat{H}_1$ (resp. $\hat{H}_2$) is an HEDCS of $G_1$ (resp. $G_2$) with parameters
$(\beta,(1 -\lambda)\beta)$ with probability $1 - \frac{4}{n^{5d-1}}$
. In the following, we condition on this event. By Lemma \ref{degree_dist} (Degree Distribution Lemma), since both $H_1$ (resp. $H_2$) and $\hat{H}_1$ (resp. $\hat{H_2}$) are
HEDCSs for $G_1$ (resp. $G_2$), the degree of vertices in both of them should be ``close" to each other.
Moreover, since by Lemma \ref{sampledhedcs} the degree of each vertex in $\hat{H}_1$ and $\hat{H}_2$ is close to $p$ times its degree
in $H$, we can argue that the vertex degrees in $H_1$ and $H_2$ are close. Formally, for any $v \in V$ , we
have

\begin{tabular}{lll}
\\
  $|d_{H_1}(v) - d_{H_2}(v)|$   & $\leq$ & $|d_{H_1}(v) - d_{\hat{H}_1}(v)| + |d_{\hat{H}_1}(v) - d_{\hat{H}_2}(v)| + |d_{\hat{H}_2}(v) - d_{H_2}(v)|$  \\
     & $\leq$ & $O\big(n^{1/2}\big) \lambda^{1/2} \beta  + | d_{\hat{H}_1}(v) - p d_H(v)| + | d_{\hat{H}_2}(v) - p d_H(v)|$\\
     & $\leq$ & $O\big(n^{1/2}\big) \lambda^{1/2} \beta  + O(1) \cdot \lambda \cdot \beta \ . $\\
\end{tabular}

\end{proof}

\begin{restatable}{corollary}{cor16}\label{cor16}
If $G$ is linear, then $|d_{H_1}(v) - d_{H_2}(v)| = O\big(\log{n}\big) \lambda \beta$.
\end{restatable}
\vspace{2mm}
We are now ready to present a parallel algorithm that will use the HEDCS subgraph. We first compute an HEDCS in parallel via edge sampling. Let  $G^{(1)}, \ldots , G^{(k)}$ be a random $k$-partition of a graph $G$. We show that
if we compute an arbitrary HEDCS of each graph $G^{(i)}$ (with no coordination across different graphs)
and combine them together, we obtain a HEDCS for the original graph $G$. We then store this HEDCS in one machine and compute a maximal matching on it. We present our algorithm for all range
of memory $s = n^{\Omega(1)}$. Lemma \ref{cishedcs} and Corollary \ref{linearmemory} serve as a proof to Theorem \ref{result4}.

\begin{algorithm}
\SetAlgoLined
\DontPrintSemicolon
 Define $k:= \frac{m}{s\log{n}}$, \ $\lambda := \frac{1}{2n\log{n}}$ and $\beta := 500\cdot d^3 \cdot n^2 \cdot \log^3{n}$.\;
 $G^{(1)}, \ldots, G^{(k)}$:= random $k$-partition of $G$.\;
 \For{$i = 1$ to $k$, in parallel}{
 Compute $C^{(i)} = HEDCS(G^{(i)}, \beta,(1 - \lambda) \cdot \beta)$ on machine $i$.\;
 }
 Define the multi-graph $C(V, E_C)$ with $E_C := \cup_{i=1}^k C^{(i)}$.\; 
 Compute and output a maximal matching on $C$.\;

 \caption{\textsf{HEDCS-Matching}($G, s)$: a parallel algorithm to compute a $O(d^3)$-approximation matching on a $d$-uniform hypergraph $G$ with $m$ edges on machines of memory $O(s)$}
\end{algorithm}

\begin{restatable}{lemma}{cishedcs}
\label{cishedcs}
Suppose $k \leq \sqrt{m}$. Then with high probability 
\begin{enumerate}

\item The subgraph $C$ is an HEDCS$(G, \beta_C , \beta_C^-)$  for parameters:
$ \lambda_C = O\big(n^{1/2}\big) \lambda^{1/2}  \mbox{   ,   } \beta_C = (1 + d \cdot  \lambda_C) \cdot k \cdot \beta \mbox{   and   } \beta_C^- = (1- \lambda- d \cdot  \lambda_C) \cdot k \cdot \beta $.
    \item The total number of edges in each subgraph $G^{(i)}$ of $G$ is $\tilde{O}(s)$.
    \item If $s = \tilde{O}(n\sqrt{nm})$, then the graph 
$C$ can fit in the memory of one machine.
\end{enumerate}
\end{restatable}

\begin{proof}[Proof of Lemma \ref{cishedcs}]
$ $
\begin{enumerate}
    \item Recall that each graph $G^{(i)}$
is an edge sampled subgraph of $G$ with sampling probability $p = \frac{1}{k}$. By Lemma \ref{lemma25} for graphs $G^{(i)}$ and $G^{(j)}$
(for $i \neq j \in [k])$ and their HEDCSs $C^{(i)}$ and $C^{(j)}$,
with probability $1 - \frac{4}{n^{5d-1}}$
, for all vertices $v \in V$ :

$$|d_{C^{(i)}}(v) - d_{C^{(j)}}(v)| \leq O\big(n^{1/2}\big) \lambda^{1/2} \beta \ .$$

By taking a union bound on all ${k}\choose{2}$ $\leq n^d$ pairs of subgraphs $G^{(i)}$ and $G^{(j)}$
for $i \neq j \in [k]$, the above
property holds for all $i, j \in [k]$, with probability at least $1 - \frac{4}{n^{4d-1}}$. In the following, we condition
on this event.

We now prove that $C$ is indeed a $HEDCS(G, \beta_C , \beta_C^-)$. First, consider an edge $e \in C$ and let
$j \in [k]$ be such that $e \in C^{(j)}$ as well. We have
\begin{eqnarray*}
\sum\limits_{v \in e} d_C(v) & = & \sum\limits_{v \in e} \sum\limits_{i=1}^k d_{C^{(i)}}(v) \\
& \leq  & k \cdot \sum\limits_{v \in e} d_{C^{(j)}}(v) + d \cdot k \cdot \lambda_C \cdot \beta \\ 
& \leq &  k \cdot \beta + d \cdot k \cdot \lambda_C \cdot \beta \\
& = & \beta_ C \ .
\end{eqnarray*}

Hence, $C$ satisfies Property \textit{(P1)} of HEDCS for parameter $\beta_C$.
Now consider an edge $e \in G \setminus C$ and let $j \in [k]$ be such that $e \in G^{(j)} \setminus C^{(j)}$
(recall that
each edge in $G$ is sent to exactly one graph $G^{(j)}$ in the random $k$-partition). We have,
\begin{eqnarray*}
\sum\limits_{v \in e} d_C(v) &  = & \sum\limits_{v \in e} \sum\limits_{i=1}^k d_{C^{(i)}}(v) \\
& \geq &  k \cdot \sum\limits_{v \in e} d_{C^{(j)}} - d \cdot k \lambda_C \beta \\
& \geq & k \cdot (1- \lambda) \cdot \beta - d \cdot k \lambda_C \beta  \ . 
\end{eqnarray*}
\item Let $E^{(i)}$ be the edges of $G^{(i)}$. By the independent sampling of edges in an edge sampled subgraph, we have that $E \big[|E^{(i)}|\big] = \tfrac{m}{k} = \tilde{O}(s)$. By Chernoff bound, with
probability $1 - \tfrac{1}{k\cdot n^{20}}$, the size of $E^{(i)}$ is $\tilde{O}(s)$ . We can then take a union bound on all $k$ machines
in $G^{(i)}$ and have that with probability $1 - 1/ n^{20}$, each graph  $G^{(i)}$
is of size $\tilde{O}(s)$.
\item The number of edges in $C$ is bounded by $n \cdot \beta_c = O(n \cdot k\cdot \beta) = \tilde{O}(\frac{n^3m}{s}) =  \tilde{O}(s)$.
\end{enumerate}
\end{proof}

\begin{restatable}{corollary}{linearmemory}\label{linearmemory}
If $G$ is linear, then by choosing $\lambda:= \frac{1}{2\log^2{n}}$ and $\beta:= 500 \cdot d^3 \cdot \log^4{n} \ $ in \textsf{HEDCS-Matching} we have:
\begin{enumerate}
    \item With high probability,  the subgraph $C$ is a HEDCS$(G, \beta_C , \beta_C^-)$  for parameters:
$ \lambda_C = O(\log{n}) \lambda  \mbox{   ,   } \beta_C = (1 + d \cdot  \lambda_C) \cdot k \cdot \beta \mbox{   and   } \beta_C^- = (1- \lambda- d \cdot  \lambda_C) \cdot k \cdot \beta $.
    \item If $s= \tilde{O}(\sqrt{nm})$ then $C$ can fit on the memory of one machine.
\end{enumerate}
\end{restatable}

\begin{proof}[Proof of Corollary \ref{linearmemory}]
$ $
\begin{enumerate}
    \item Similarly to Lemma \ref{cishedcs} and by using corollary \ref{cor16}, we know that for graphs $G^{(i)}$ and $G^{(j)}$
(for $i \neq j \in [k])$ and their HEDCSs $C^{(i)}$ and $C^{(j)}$,
with high probability , for all vertices $v \in V$ :

$$|d_{C^{(i)}}(v) - d_{C^{(j)}}(v)| \leq O\big(\log{n}\big) \lambda \beta. $$

By taking a union bound on all ${k}\choose{2}$ $\leq n^d$ pairs of subgraphs $G^{(i)}$ and $G^{(j)}$
for $i \neq j \in [k]$, the above
property holds for all $i, j \in [k]$, with probability at least $1 - \frac{4}{n^{4d-1}}$. In the following, we condition
on this event. Showing that $C$ is indeed a $HEDCS(G, \beta_C , \beta_C^-)$ follows by the same analysis from the proof of Lemma \ref{cishedcs}.
\item The number of edges in $C$ is bounded by $n \cdot \beta_c = O(n \cdot k\cdot \beta) = O(n \cdot k) = \Tilde{O}(\frac{nm}{s}) = \tilde{O}(s)$.
\end{enumerate}
\end{proof}

The previous lemmas allow us to formulate the following theorem.

\begin{restatable}{theorem}{result4}
\label{result4}
\textsf{HEDCS-Matching }constructs a HEDCS of $G$ in 3 $MPC$ rounds on machines of memory $s = \tilde{O}(n\sqrt{nm})$ in general and $s = \tilde{O}(\sqrt{nm})$ for linear hypergraphs.
\end{restatable}

\begin{corollary}\label{result5}
\textsf{HEDCS-Matching} achieves a $d(d-1+1/d)^2$-approximation to the $d$-Uniform Hypergraph Matching in 3 rounds with high probability.
\end{corollary}

\begin{proof}[Proof of Corollary \ref{result5}]
We show that with high probability, $C$ verifies the assumptions of theorem \ref{mainth}. From Lemma \ref{cishedcs}, we get that with high probability, the subgraph $C$ is a $HEDCS(G, \beta_C , \beta_C^-)$  for parameters:
$ \lambda_C = O\big(n^{1/2}\big) \lambda^{1/2}  \mbox{   ,   } \beta_C = (1 + d \cdot  \lambda_C) \cdot k \cdot \beta \mbox{   and   } \beta_C^- = (1- \lambda- d \cdot  \lambda_C) \cdot k \cdot \beta.$
We can see that $\beta_c \geq \frac{8 d^2}{d-1} \cdot \lambda_c^{-3}$. Therefore by Theorem \ref{mainth}, $C$ contains a $(d-1+\frac{1}{d})$-approximate $\epsilon$-restricted matching. Since the integrality gap of the $d$-UHM is at most $d-1+\frac{1}{d}$ (see \cite{chan2012linear} for details), then $C$ contains a $(d-1+\frac{1}{d})^2$-approximate matching. Taking a maximal matching in $C$ multiplies the approximation factor by at most $d$. Therefore, any maximal matching in $C$ is a $d(d-1+\frac{1}{d})^2$-approximation.
\end{proof}
\section{Computational Experiments}
\label{sec:experiments}
To understand the relative performance of the proposed algorithms, we conduct a wide variety of experiments on both random and real-life data \cite{sen2008collective,yang2015defining, kunegis2013konect}.  We implement the three algorithms \textsf{Greedy}, \textsf{Iterated-Sampling} and \textsf{HEDCS-Matching} using Python, and more specifically relying on the module \textit{pygraph} and its class \textit{pygraph.hypergraph} to construct and perform operations on hypergraphs. 
We simulate the MPC model by computing a $k$-partitioning and splitting it into $k$ different inputs. Parallel computations on different parts of the $k$-partitioning are handled through the use of the \textit{multiprocessing} library in Python. We compute the optimal matching through an Integer Program for small instances of random uniform hypergraphs ($d\leq 10)$ as well as geometric hypergraphs. The experiments were conducted on a 2.6 GHz Intel Core i7 processor and 16 GB RAM workstation. The datasets differ in their number of vertices, hyperedges, vertex
degree and hyperedge cardinality.  In the following tables, $n$ and $m$ denote the number of vertices and number of hyperedges respectively, and $d$ is the size of hyperedges. For Table 3, the graphs might have different number of edges and $\Bar{m}$ denotes the average number of edges. $k$ is the number of machines used to distribute the hypergraph initially. We limit the number of edges that a machine can store to $\frac{2m}{k}$. In the columns Gr, IS and HEDCS, we store the average ratio between the size of the matching computed by the algorithms \textsf{Greedy}, \textsf{Iterated-Sampling} and \textsf{HEDCS-Matching} respectively, and the size of a benchmark. This ratio is computed by the percentage $\frac{ALG}{BENCHMARK}$, where $ALG$ is the output of the algorithm, and $BENCHMARK$ denotes the size of the benchmark.  The benchmarks include the size of optimal solution when it is possible to compute, or the size of a maximal matching computed via a sequential algorithm. $\#I$ denotes the number of instances of random graphs that we generated for fixed $n$, $m$ and $d$. $\beta$ and $\beta^-$ are the parameters used to construct the HEDCS subgraphs in \textsf{HEDCS-Matching}. These subgraphs are constructed using the procedure in the proof of Lemma \ref{existence}.

\subsection{Experiments with random hypergraphs}
We perform experiments on two classes of $d$-uniform random hypergraphs. The first contains random uniform hypergraphs, and the second contains random geometric hypergraphs.\\

\noindent \textbf{Random Uniform Hypergraphs.}
For a fixed $n$, $m$ and $d$, each potential hyperedge is sampled independently and uniformly at random from the set of vertices. In Table 2, we use the size of a perfect matching $\frac{n}{d}$ as a benchmark, because a perfect matching in random graphs exists with probability $1-o(1)$ under some conditions on $m,n$ and $d$. If $d(n,m) = \frac{m\cdot d}{n}$ is the expected degree of a random uniform hypergraph, Frieze
and Janson \cite{frieze1995perfect} showed that $\frac{d(n,m)}{n^{1/3}} \rightarrow \infty$ is a sufficient condition for the hypergraph to have a perfect matching with high probability. Kim \cite{kim2003perfect} further weakened this condition to $\frac{d(n,m)}{n^{1/(5 + 2/(d-1))}}\rightarrow \infty$. We empirically verify, by solving the IP formulation, that for $d=3,5$ and for small instances of $d=10$, our random graphs contain a perfect matching.
In Table 2, the benchmark is the size of a perfect matching, while in Table 5, it is the size of a greedy maximal matching. In terms of solution quality (Tables 2 and 5) \textsf{HEDCS-Matching} performs consistently better than \textsf{Greedy}, and \textsf{Iterated-Sampling} performs significantly better than the other two. None of the three algorithms are capable of finding  perfect matchings for a significant number of the runs.
When compared to the size of a maximal matching, \textsf{Iterated-Sampling} still performs better, followed by \textsf{HEDCS-Matching}. However, the ratio is smaller when compared to a maximal matching, which is explained by the deterioration of the quality of greedy maximal matching as $n$ and $d$ grow. Dufosse \textit{et al}. \cite{dufosse2019effective} confirm that the approximation quality of a greedy maximal matching on random graphs that contain a perfect matching degrades as a function of $n$ and $d$. The performance of the algorithms decreases as $d$ grows, which is theoretically expected since their approximations ratio are both proportional to $d$. The number of rounds for \textsf{Iterated-Sampling}  grows slowly with $n$, which is consistent with $O(\log{n})$ bound. Recall that the number of rounds for the other two algorithms is constant and equal to 3.\\

\noindent\textbf{Geometric Random Hypergraphs.} The second class we experimented on is random geometric hypergraphs. The vertices of a random geometric hypergraph (RGH) are randomly sampled from the uniform distribution of the space $[0,1)^2$. A set of $d$ different vertices $v_1, \ldots, v_d \in V$ forms a hyperedge if, and only if, the distance between any $v_i$ and $v_j$ is less than a previously specified parameter $r \in (0,1)$. The parameters $r$ and $n$ fully characterize a RGH. We fix $d=3$ and generate different geometric hypergraphs by varying $n$ and $r$. 
We compare the output of the algorithms to the optimal solution that we compute through the IP formulation. Table 3 shows that the performance of our three algorithms is almost similar with \textsf{Iterated-Sampling} outperforming \textsf{Greedy} and \textsf{HEDCS-Matching} as the size of the graphs grows. We also observe that random geometric hypergraphs do not contain perfect matchings, mainly because of the existence of some outlier vertices that do not belong to any edge. The number of rounds of \textsf{Iterated-Sampling} still grows with $n$, confirming the theoretical bound and the results on random uniform hypergraphs.

{\small
\begin{center}

\begin{tabular}{|c|c|c|c|c|c|c|c|c|c|c|}
  \hline
  $n$ & $m$ & $d$ & $k$ & \#I & Gr & IS & HEDCS & $\beta$ &  $\beta^-$ & Rounds IS \\
  \hline
  \hline
  15 & 200 & \multirow{4}{*}{3} & 5 & \multirow{4}{*}{500} & 77.6\% & 86.6\% & 82.8\% & 5 &  3 & 3.8 \\ 
  30 & 400 &  & 5 &  & 78.9\% & 88.1\%  & 80.3\%  & 7 & 4 & 4.56\\
  100 & 3200 &  & 10 &  & 81.7\% & 93.4\% & 83.1\% & 5 & 2 & 5.08   \\
  300 & 4000 &  & 10 &  & 78.8\% & 88.7\% & 80.3\% & 8 & 6 & 7.05 \\
  \hline
  50 & 800 & \multirow{4}{*}{5} & 6 & \multirow{4}{*}{500} & 66.0\% & 76.2\% &   67.0\% & 16 & 11 & 4.89  \\
  100 & 2,800 &  & 10 &  & 68.0\% & 79.6\% & 69.8\%  & 16 & 11 & 4.74    \\
  300 & 4,000 &  & 10 &  & 62.2\% & 75.1\% &  65.5\% & 10 & 5 & 6.62 \\
 500 & 8,000 &  & 16 &  & 63.3\% & 76.4\% & 65.6\% & 10 & 5 & 7.62\\
 \hline
  500 & 15,000 & \multirow{3}{*}{10} & 16 & \multirow{4}{*}{500} & 44.9\% & 58.3\% & 53.9\% & 20 & 10 & 6.69   \\
  1,000 & 50,000 &  & 20 &  & 47.3\% & 61.3\% & 50.5\%  & 20 & 10 & 8.25  \\
  2,500 & 100,000 &  & 20 &  & 45.6\% & 59.9\% & 48.2\%  & 20 & 10 & 8.11 \\
  5,000 & 200,000 &  & 20 &  & 45.0\% & 59.7\% &  47.8\% & 20 & 10 & 7.89 \\
  \hline
   1,000 & 50,000 & \multirow{4}{*}{25} & 25 & \multirow{4}{*}{100} & 27.5\% & 34.9\% & 30.8\%  & 75 & 50 & 8.10  \\
   2,500 & 100,000 &  & 25 &  & 26.9\% & 34.0\% & 27.0\%  & 75 & 50 & 8.26   \\
   5,000 & 250,000 &  & 30  &  & 26.7\% & 33.8\% & 28.8\%  & 75 & 50 & 8.23 \\
   10,000 & 500,000 &  & 30 &  & 26.6\% & 34.1\% &  28.2\%  & 75 & 50 &  8.46\\
  \hline
  5,000 & 250,000 & \multirow{4}{*}{50} & 30 & \multirow{4}{*}{100} & 22.4\% & 30.9\% &  27.9\%  & 100 & 50 & 10.22\\
  10,000 & 500,000 & & 30 &  & 22.2\% & 31.0\% & 26.5\%   & 100 & 50 &  10.15 \\
  15,000 & 750,000 &  & 30 &  & 20.9\% & 30.8\% &   26.4\%  & 100 & 50 & 10.26  \\
  25,000 & 1,000,000 & & 30 &  & 20.9\% & 30.8\% & 26.4\%   & 100 & 50 & 10.29 \\
  \hline
\end{tabular}
$ $
\captionof{table}{Comparison on random instances with perfect matching benchmark, of size $\frac{n}{d}$.} 

\end{center}}

{\small
\begin{center}

\begin{tabular}{|c|c|c|c|c|c|c|c|c|c|c|c|}
  \hline
  $n$ & $r$ & $\Bar{m}$ & $d$ & $k$ & \#I & Gr & IS & HEDCS & $\beta$ &  $\beta^-$ & Rounds IS  \\
  \hline
  \hline
  100 & 0.2 & 930.1 $\pm$ 323 & \multirow{4}{*}{3} & 5 & \multirow{4}{*}{100} & 88.3\% & 89.0\% & 89.6\% &
  3 & 5 & 4.1\\ 
  100 & 0.25 & 1329.5 $\pm$ 445 &  & 10 &  & 88.0\% & 89.0\% & 89.5\% & 3 & 5 & 5.2 \\ 
  250 & 0.15 & 13222$\pm$ 3541 &  & 5 &  & 85.0\% & 88.6\% & 85.5\% & 4 & 7 & 8.1  \\
  300 & 0.15 & 14838  $\pm$ 4813 & & 10 &  & 85.5\% & 88.0\%& 85.2\% & 4 & 7 & 11.1 \\
  300 & 0.2 & 27281 $\pm$ 8234 & & 10 &  & 85.0\% & 89.0\%&  86.3\% & 4 & 7 & 13.5 \\
  \hline

\end{tabular}
$ $
\captionof{table}{Comparison on random geometric hypergraphs with optimal matching benchmark.} 
\end{center}
}

\subsection{Experiments with real data}

\textbf{PubMed and Cora Datasets.}
We employ two citation network datasets, namely the Cora and Pubmed datasets \cite{sen2008collective,yang2015defining}. These datasets are represented with a graph, with vertices being publications, and edges being a citation links from one article to another. We construct the hypergraphs in two steps 1) each article is a vertex; 2) each article is taken as a centroid and forms a hyperedge to connect those articles which have citation links to it (either citing it or being cited). The Cora hypergraph has an average edge size of $3.0 \pm 1.1$, while the average in the Pubmed hypergraph is $4.3 \pm 5.7$. The number of edges in both is significantly smaller than the number of vertices, therefore we allow each machine to store only $\frac{m}{k} + \frac{1}{4}\frac{m}{k}$ edges. We randomly split the edges on each machine, and because the subgraphs are small, we are able to compute the optimal matchings on each machine, as well as on the whole hypergraphs. We perform ten runs of each algorithm with different random $k$-partitioning and take the maximum cardinality obtained. Table 4 shows that none of the algorithms is able to retrieve the optimal matching. This behaviour can be explained by the loss of information that using parallel machines implies. We see that \textsf{Iterated-Sampling}, like in previous experiments, outperforms the other algorithms due to highly sequential design. \textsf{HEDCS-Matching} particularly performs worse than the other algorithms, mainly because it fails to construct sufficiently large HEDCSs.\\ 

\noindent\textbf{Social Network Communities.} We include two larger real-world datasets,
orkut-groups and LiveJournal, from the Koblenz Network Collection \cite{kunegis2013konect}. We use two hypergraphs that were constructed from these datasets by Shun \cite{shun2020practical}. Vertices represent individual users, and hyperedges represent communities in the network. Because membership in these communities does
not require the same commitment as collaborating on academic research, these hypergraphs have different
characteristics from co-citation hypergraphs, in terms of size, vertex degree and hyperedge cardinality. We use the size of a maximal matching as a benchmark. Table 4 shows that \textbf{Iterated-Sampling} still provides the best approximation. \textbf{HEDCS-Sampling} performs worse than \textbf{Greedy} on Livejournal, mainly because the ratio $\frac{m}{n}$ is not big enough to construct an HEDCS with a large matching.
 
{\small
\begin{center}
\begin{tabular}{|c|c|c|c|c|c|c|c|}
  \hline
  Name & $n$ & $m$  & $k$ &  Gr & IS & HEDCS &  Rounds IS\\
  \hline
  \hline
   Cora & 2,708 & 1,579 & 2 & 75.0\% & 83.2\% & 63.9\% &  6   \\
  \hline
  PubMed & 19,717 & 7,963 & 3 & 72.0\% & 86.6\% &  62.4\% & 9 \\
  \hline
  Orkut & 2,32 $\times 10^6$ & 1,53 $\times 10^7$  & 10 & 55.6\%  & 66.1\% & 58.1\% & 11  \\
  \hline
  Livejournal & 3,20 $\times 10^6$ & 7,49 $\times 10^6$ &  3 & 44.0\%  &  55.2\% & 43.3\% & 10 \\
  \hline
\end{tabular}
\captionof{table}{Comparison on co-citation and social network hypergraphs.} 
\end{center}}

\subsection{Experimental conclusions}
In the majority of our experiments, \textsf{Iterated-Sampling} provides the best approximation to the $d$-UHM problem, which is consistent with its theoretical superiority. On random graphs, \textsf{HEDCS-Matching} performs consistently better than \textsf{Greedy}, even-though \textsf{Greedy} has a better theoretical approximation ratio. We suspect it is because the $O(d^3)$-approximation bound on \textsf{HEDCS-Matching} is loose. We conjecture that rounding an $\epsilon-$restricted matching can be done efficiently, which would  improve the approximation ratio. The performance of the three algorithms decreases as $d$ grows.
    The results on the number of rounds of \textsf{Iterated-Sampling} also are consistent with the theoretical bound.  However, due to its sequential design, and by centralizing the computation on one single machine while using the other machines simply to coordinate, \textsf{Iterated-Sampling} not only takes more rounds than the other two algorithms, but is also slower when we account for the absolute runtime as well as the runtime per round. We compared the runtimes of our three algorithms on a set of random uniform hypergraphs. Figure \ref{fig:runtime} and \ref{fig:runtime_per_round} show that the absolute and per round run-times of \textsf{Iterated-Sampling} grow considerably faster with the size of the hypergraphs. We can also see that \textsf{HEDCS-Sampling} is slower than \textsf{Greedy}, since the former performs heavier computation on each machine. This confirms the trade-off between the extent to which the algorithms use the power of parallelism and the quality of the approximations.


\section{Conclusion}
We have presented the first algorithms for the $d$-UHM problem in the MPC model.
Our theoretical and experimental results highlight the trade-off between the approximation ratio, the necessary memory per machine and the number of rounds it takes to run the algorithm. 
We have also introduced the notion of HEDCS subgraphs, 
and have shown that an HEDCS contains a good approximation for the maximum matching and that they can be constructed in few rounds in the MPC model. 
We believe better approximation algorithms should be possible, especially if we can give better rounding algorithms for a $\epsilon$-restricted fractional hypergraph matching. For future work, it would be interesting to explore whether we can achieve better-than-$d$ approximation in the MPC model in a polylogarithmic number of rounds. Exploring algorithms relying on vertex sampling instead of edge sampling might be a good candidate. In addition, our analysis in this paper is specific to unweighted hypergraphs, and we would like to extend this to weighted hypergraphs. 

\begin{figure}[!h]
    \centering
    \includegraphics[scale=0.45]{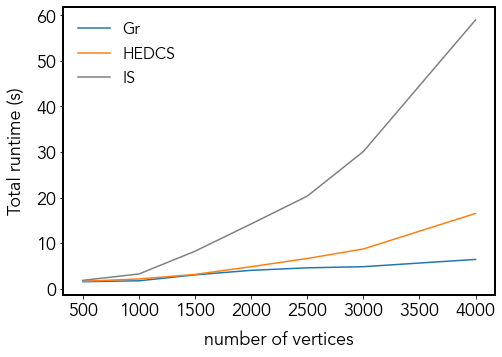}
    \caption{Runtime of the three algorithms when $d=3$ and $m = 20\cdot d \cdot n$}
    \label{fig:runtime}
\end{figure}

\begin{figure}[!h]
    \centering
    \includegraphics[scale=0.45]{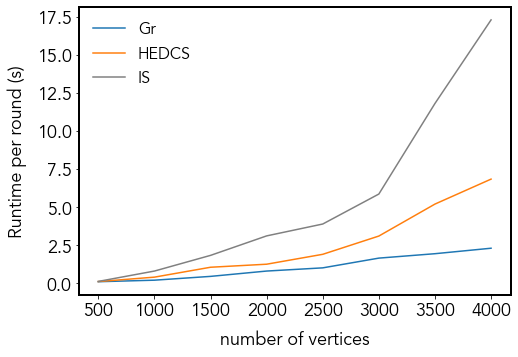}
    \caption{Total runtime per round (maximum over all machines in every round) of the three algorithms when $d=3$ and $m = 20\cdot d\cdot n$}
    \label{fig:runtime_per_round}
\end{figure}

\newpage

\bibliographystyle{abbrv}	
\bibliography{bibliography.bib}

\newpage
\appendix
\section{Omitted proofs}

\subsection{Proof of Lemma \ref{lemmaval}}

\lemmaval*

\begin{proof} 
We distinguish three cases:
\begin{itemize}
    \item 
$d_H(x) \leq \frac{\beta}{d}$: in this case $\frac{1}{\beta - d_H(x)} \leq \frac{d}{(d - 1) \beta} < \epsilon$ and $d_{H^*}(x) \leq d_H(x) \leq \beta - d_H(x) $. This implies that $val'(x) = \frac{d_{H^*}(x)}{\beta - d_H(x)}$, so that :
$$ E[val'(x)] \geq \frac{d-1}{d}\cdot \frac{d_{H}(x)}{\beta - d_H(x)} = \phi(d_H(x))\ .$$ Now consider the case in which $d_H(x) > \frac{\beta}{d}$. Then
$$ E[d_{H^*}(x)] \geq \frac{(d-1)d_H(x)}{d} > \frac{(d-1)}{d^2} \cdot \beta \geq 8 \cdot \lambda^{-3} >> \epsilon^{-1} \ ,$$
because $\beta \geq \frac{8 d^2}{d-1} \cdot \lambda^{-3}$, and by a Chernoff Bound : 
\begin{equation}
\label{chernoff}
\begin{split}
P\Big[ d_{H^*}(x) < (1 - \frac{\lambda}{2}) \frac{(d-1)d_H(x)}{d} \Big] & \leq \exp{(-E[d_{H^*}(x)] (\frac{\lambda}{2})^2 \frac{1}{2})}\\
& \leq \exp{(-\lambda^{-1})}\\ &\leq \lambda /2 \ .
\end{split}
\end{equation}

\item Let us now consider the case $d_H(x)> \frac{\beta}{d}$ and $\min\Big\{ \epsilon, \frac{1}{\beta - d_H(x)} \Big\} = \epsilon$. With probability at least $(1 - \frac{\lambda}{2})$, we have that:
$$ d_{H^*}(x) \geq (\beta - \frac{1}{\epsilon}) (1 - \frac{\lambda}{2}) \frac{d-1}{d} >> \epsilon^{-1}.$$
Thus with probability at least $(1 - \frac{\lambda}{2})$, we have that $d_{H^*}(x) \epsilon > 1$ and: $$E[val'(x)] \geq (1 - \frac{\lambda}{2}) \geq (1 - \frac{\lambda}{2}) \phi(x).$$ 

\item The only case that we need to check is $d_H(x) \geq \frac{\beta}{d}$ and $\min\Big\{ \epsilon, \frac{1}{\beta - d_H(x)} \Big\} =\frac{1}{\beta - d_H(x)}$, so that $val'(x) = \min\Big\{ 1, \frac{d_{H^*}(x)}{\beta - d_H(x)} \Big\}$. Again we have that with probability at least $(1 - \frac{\lambda}{2}):$
$$ \frac{d_{H^*}(x)}{\beta - d_H(x)} \geq \frac{d-1}{d} \frac{d_{H}(x)}{\beta - d_H(x)} (1 - \frac{\lambda}{2}) \geq (1 - \frac{\lambda}{2}) \phi(d_H(x)).$$ 

\end{itemize}
In other words, with probability at least $(1 - \frac{\lambda}{2})$, we have $val(x) \geq (1 - \frac{\lambda}{2}) \phi(d_H(x))$, so that $E[val(x)] \geq (1 - \frac{\lambda}{2})^2 \phi(d_H(x)) > (1 -\lambda) \phi( d_H(x))$. We just showed that in all cases $E[val'(x)] \geq (1 -\lambda) \phi(d_H(x))$
\end{proof}

\subsection{Proof of Lemma \ref{phi} and Lemma \ref{4prop}}\label{appmainth}

\phix*

\begin{proof}
We will provide a proof for $d=3$ that is easy to generalize. We first show that if $a + b + c \geq \beta$, then $\phi(a) + \phi(b) + \phi (c) \geq 1$. The claim is true if $\phi(a) \geq 1$ or $\phi(b) \geq 1$ or $\phi(c) \geq 1$. Suppose that $\phi(a) < 1$ and $\phi(b) < 1$ and $\phi(c) < 1$. Then :
$$ \phi(a) + \phi(b) + \phi (c) = \frac{d-1}{d} \Big( \frac{a}{\beta - a} + \frac{b}{\beta - b} + \frac{c}{\beta - c}\Big) \geq \frac{d-1}{d} \Big( \frac{a}{b+c} + \frac{b}{a+c} + \frac{c}{a+b}\Big) \ . $$
By Nesbitt's Inequality we know that
$$  \frac{a}{b+c} + \frac{b}{a+c} + \frac{c}{a+b} \geq \frac{d}{d-1} = \frac{3}{2},$$

and therefore $\phi(a) + \phi(b) + \phi (c) \geq 1$.\\

By the general Nesbitt's Inequality (See appendix \ref{nesbitt}), we know that for $d>3$
$$ \sum\limits_{i = 1}^d \frac{a_i}{\sum\limits_{j\neq i} a_j} \geq \frac{d}{d-1}\ .$$
So if $\sum\limits_{i=1}^d a_i \geq \beta$, then $\sum\limits_{i=1}^d \phi(a_i) \geq 1$.
Now, let $\phi'(x) = \frac{d}{dx}\phi(x)$. To complete the proof, it is sufficient to show that we always have $\phi'(x) \leq \frac{5}{\beta}$. To prove this inequality, note that if $x \geq \frac{d}{2d-1}\beta$ then $\phi(x) = 1$ and thus $\phi'(x) = 0.$ Now, if $x \leq \frac{d}{2d-1}\beta$  then: 
$$ \phi'(x) = \frac{d-1}{d}\frac{d}{dx} \frac{x}{\beta - x} = \frac{d-1}{d} \frac{\beta}{(\beta - x)^2}\ ,$$
which is increasing in $x$ and maximized at $x = \frac{d}{2d-1}\beta$, in which case $\phi'(x) =  \frac{(2d-1)^2}{d(d-1)} \frac{1}{\beta} \leq \frac{5}{\beta}$. In the end we get:
$$\sum\limits_{i=1}^d \phi(a_i) \geq 1 - 5\cdot \lambda\ .$$
\end{proof} 

\fourprop*

\begin{proof}

Let $M^G$ be some maximum integral matching in $G$. Some of the edges in $M^G$ are in $H$, while others are in $G \setminus H$. Let $X_0$ contain all vertices incident to edges in
$M_G \cap (G\setminus H)$, and let $Y_0$ contain all vertices incident to
edges in $M_G \cap H$. We now show that $X_0$ and $Y_0$ satisfy
the first three properties of the lemma. Property 1 is
satisfied because $X_0 \cup Y_0$ consists of all matched vertices
in $M_G$. Property 2 is satisfied by definition of $Y_0$. To
see that Property 3 is satisfied, remark that the vertices of $Y_0$ each contribute exactly $\frac{1}{d}$. Now, $X_0$ consists of 
$|X_0|/d$ disjoint edge in $G \setminus H$, and
by Property P2 of a HEDCS, for each such edge $e$ : $\sum\limits_{x \in e} d_H(x) \geq \beta (1-\lambda)$ 
and by Lemma \ref{phi}, we have $\sum\limits_{x \in e} \phi(d_H(x)) \geq (1-5\lambda)$ and each one of these vertices contributes in average at least $\frac{1-5\lambda}{d} $ to $\sigma$, just as desired. Loosely
speaking, $\phi(d_H(x))$ will end up corresponding to the
profit gained by vertex $x$ in the fractional matching $M^H_f$.\\

Consider $Y_0 $ and $X_0$ from above. These sets might not satisfy the Property 4 (that all edges in $H$ with an endpoint in $X$ have at least one other endpoint in $Y$). Can we transform these into sets $X_1$ and $Y_1$, such that the first three properties still hold and there are no hyperedges with endpoints in $X_1$ and $V\setminus (X_1 \cup Y_1)$; at this stage, however, there will be possibly edges in $H$ with different endpoints in $X_1$. To construct $X_1$,$Y_1$, we start with $X = X_0$ and $Y = Y_0$, and present a transformation that terminates with $X = X_1$ and $Y = Y_1$. Recall that $X_0$ has a perfect matching using edges in $G\setminus H$. The set $X$ will maintain this property throughout the transformation, and each vertex $x \in X$ has always a unique \textit{mate} $e'$. The construction does the following : as long as there exists an edge $e$ in $H$ containing $x$ and only endpoints in $ X$ and  $V \setminus (X \cup Y)$, let $e'$ be the mate of $x$, we then remove the endpoints of $e'$ from $X$ and add the endpoints of $e$ to $Y$. Property 1 is maintained because we have removed $d$ vertices from $X$ and added $d$ to $Y$. Property 2 is maintained because the vertices we added to $Y$ were connected by an edge in $H$. Property 3 is maintained because $X$ clearly still has a perfect matching in $G \setminus H$, and for the vertices $\{x_1, x_2, \ldots, x_d\} = e'$, the average contribution is still at least $\frac{1 - 5\lambda}{d}$, as above. We continue this process while there is an edge with endpoints in $X$ and $V \setminus (X \cup Y)$. The process terminates because each time we are removing $d$ vertices from $X$ and adding $d$ vertices to $Y$. We end up with two sets $X_1$ and $Y_1$ such that the first three properties of the lemma are satisfied and there are no edges with endpoints in $X_1$ and $V \setminus (X_1 \cup Y_1)$. This means that for any edges in $H$ incident to $X$, this edge is either incident to $Y$ as well or incident to only points in $X$.\\ 

We now set $X = X_1$ and $Y = Y_1$ and show how to transform $X$ and $Y$ into two sets that satisfy all four properties of the lemma. Recall that $X_1$ still contains a perfect matching using edges in $G\setminus H$; denote this matching $M^G_X$. Our final set, however, will not guarantee such a perfect matching. Let $M^H_X$ be a maximal matching in $X$ using edges in $H$ (with edges not incident to $Y$, because they already satisfy Property 4). Consider the edge set $E^*_X = M^G_X \cup M^H_X$. \\

\begin{figure}[!h]
    \centering
    \includegraphics[scale = 0.5]{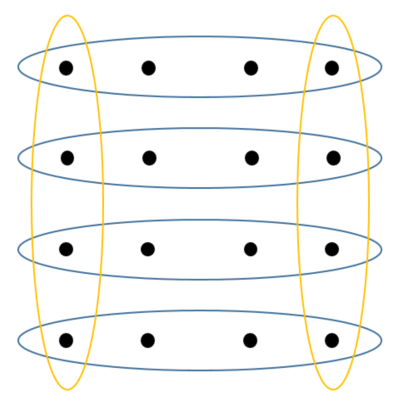}
    \caption{Example with $d = 4$. In blue the edges of $M^G_X$ and in yellow the edges of $M^H_X$}
    \label{fig:my_label}
\end{figure}

We now perform the following simple transformation, we remove the endpoints of the edges in $M^H_X$ from $X$ and add them directly to $Y$. Property 1 is preserved because we are deleting and adding the same number of vertices from $X$ and to $Y$ respectively. Property 2 is preserved because the endpoints we add to $Y$ are matched in $M^H_X$ and thus in $H$. We will see later for Property 3.\\

Let's check if Property 4 is preserved. To see this, note that we took $M_H^X$ to be maximal among edges of $H$ that contain only endpoints in $X$, and moved all the matched vertices in $M^H_X$ to $Y$. Thus all vertices that remain in $X$ are free in $M_H^X $, and so there are no edges between only endpoints in $X$ after the transformation. There are also no edges in $H$ containing endpoints from $X$ and $V \setminus (X \cup Y)$, because originally they don't exist for $X=X_1$\\

Let's check if the Property 3 is preserved. As before, this involves showing that after the transformation, the average contribution of a vertex in $X \cup Y$ to $\sigma$ is at least $\frac{1- 5\lambda}{d}$.
(Because every vertex in $X$ is incident to an edge in $E^*_X$, each vertex is accounted for in the transformation.) Now, all vertices that were in $Y_1$ remain in $Y$ , so
their average contribution remains at $1/d$. We thus
need to show that the average contribution to $\sigma$ among
vertices in $X_1$ remains at least $\frac{1- 5\lambda}{d}$ after the
transformation.\\

Let $n = |M^G_X|$ and $k = |M^H_X| $. Since $M^G_X$ is a perfect matching on $X_1$, we always have $k \leq n$. If $n = k$, then all vertices of $X_1$ are transferred to $Y$ and clearly their average contribution to $\sigma$ is $\frac{1}{d} \geq \frac{1- 5\lambda}{d}$. Now consider when $k \leq n-1$. Let the edges of $M^H_X$ be $\{e_1, \ldots e_k\}$ and these of $M^G_X$ be $\{e'_1, \ldots e'_n\}$. Let $X'$ be the set of vertices that remain in $X_1$. Because the edges of $M^G_X$ are not $H$, the by two properties of HEDCS :

\begin{eqnarray*}
\sum\limits_{1 \leq i\leq n, \ x \in e'_i} d_H(x) & \geq &  n\beta(1-\lambda) \ , \mbox{ and } \\
\sum\limits_{1 \leq i\leq k, \ x \in e_i} d_H(x) & \leq &  k\beta \ .
\end{eqnarray*}
The sum of the degrees of vertices in $X'$ can be written as the following difference:

    \begin{eqnarray}
        \sum\limits_{ x \in X'} d_H(x) & =  &\sum\limits_{1 \leq i\leq n, \ x \in e'_i} d_H(x) - \sum\limits_{1 \leq i\leq k, \ x \in e_i} d_H(x) \\ \nonumber
        &\geq & n\beta(1-\lambda) - k\beta\\ \nonumber
        & =  & (n-k)\cdot \beta - n\beta \lambda \ . \label{xprime}
    \end{eqnarray}

Now we prove that (\ref{xprime}) implies that the contribution of vertices from $X'$ on average at least $\frac{1-5\lambda}{d}$.
\begin{claim}\label{claim22}
$\sum\limits_{x \in X'} \phi(d_H(x)) \geq (n-k) - 5n\cdot\lambda$.
\end{claim}
By claim \ref{claim22}, the average contribution among the vertices that are considered is 
$$ \frac{(n-k) -5n\cdot\lambda + k }{nd} = \frac{1 - 5\lambda}{d} \ , $$ which proves Property 3 and completes the proof of the lemma.\end{proof}

\begin{proof}[Proof of claim \ref{claim22}]

Let's denote $m := n-k$. Recall that the number of vertices in $X'$ is equal to $md$ and $\phi(x) = \frac{d-1}{d} \frac{x}{\beta - x }$. Here we will prove it for $\lambda = 0$. This means that we will prove that $$ \sum\limits_{ x \in X'} d_H(x) \geq m\cdot \beta \Rightarrow \sum\limits_{x \in X'} \phi(d_H(x)) \geq m\ .$$

If $ m = n-k = 1$, then the result clearly holds by Lemma \ref{phi}. Let's suppose $k < n-1$ and thus $m \geq 2$. By Lemma \ref{phi}, we know that:
\begin{equation}
\label{applemma4}    \frac{md-1}{md} \sum\limits_{x \in X'} \frac{d_H(x)}{m \cdot \beta - d_H(x)} \geq 1 \ .
\end{equation}

We also know that :
\begin{eqnarray}
\frac{m\beta - a}{\beta - a} &  = &  m + \frac{(m-1)a}{\beta - a} \ , \mbox{ and } \\    
\frac{a}{\beta - a} & = &  m \cdot \frac{a}{m\beta - a} + \frac{(m-1)a^2}{(m\beta -a)(\beta - a) } \ .\label{comb2}
\end{eqnarray}

Combining (\ref{applemma4}) and (\ref{comb2}), we get:
\begin{equation*}
\sum\limits_{x \in X'} \frac{d_H(x)}{\beta -d_H(x)} \geq \frac{m^2d}{md-1} + \sum\limits_{x \in X'}\frac{(m-1)d_H(x)^2}{(m\beta -d_H(x))(\beta - d_H(x))}\ .    
\end{equation*}
which leads to:
\begin{equation}
\label{sumphi}
  \begin{split}
    \sum\limits_{x \in X'} \phi(d_H(x)) & \geq 
       \frac{d-1}{d}\sum\limits_{x \in X'} \frac{d_H(x)}{\beta -d_H(x)}\\
       & \geq m \cdot\frac{m(d-1)}{md-1} + \frac{d-1}{d}\sum\limits_{x \in X'}\frac{(m-1)d_H(x)^2}{(m\beta -d_H(x))(\beta - d_H(x))}\\
       & \geq m +  \frac{d-1}{d}\sum\limits_{x \in X'}\frac{(m-1)d_H(x)^2}{(m\beta -d_H(x))(\beta - d_H(x))} - \frac{m-1}{md-1} \ .
    \end{split}
\end{equation}

By convexity of the function $x \mapsto \frac{x^2}{(m\beta -x)(\beta-x)}$:
\begin{equation}
  \begin{split}
\sum\limits_{x \in X'}\frac{d_H(x)^2}{(m\beta -d_H(x))(\beta - d_H(x))} & \geq md \frac{(\frac{\sum d_H(x))}{md})^2}{(m\beta - \frac{\sum d_H(x))}{md}(\beta - \frac{\sum d_H(x))}{md})} \\
& \geq \frac{m\beta}{(md-1)(d-1)} \ .
 \end{split}
\end{equation}

Where the last inequality is due to $\sum\limits_{x \in X'} d_H(x) \geq m\cdot \beta$

Therefore, when $ \beta \geq \frac{d}{2}$
the right hand side of (\ref{sumphi}) becomes:
\begin{eqnarray*}
m + \frac{d-1}{d}\sum\limits_{x \in X'}\frac{d_H(x)^2}{(m\beta -d_H(x))(\beta - d_H(x))} - \frac{1}{md-1}
& \geq  &  m + \frac{1}{md-1}\left(\frac{m\beta}{d} - 1\right)  \\
& \geq  &  m.
\end{eqnarray*}
\end{proof}

 \subsection{Proof of Corollary \ref{degree_dist_linear}}\label{section53}

\degreedistlinear*

\begin{proof}
For linear hypergraphs, assume that $d_A(v) = k\lambda\beta$ with $k = \Omega(\log{n})$ and $d_B(v)=0$. The difference in the analysis is that, for every edge $e$ belonging to the $k \lambda \beta$ edges that are incident to $v$ in $A$, we can find a new set of at least $(k-(d+1))\lambda \beta$ edges in $B \setminus A$. In fact, for such an edge $e$, every one of the $(1-\lambda)\beta$ edges that verify $\sum\limits_{u \neq v} d_B(u) \geq (1-\lambda)\beta$ intersect $e$ in exactly on vertex that is not $v$. The same goes for the subset of at least $(k-1)\lambda\beta$ that are in $B \setminus$ A, these edges already intersect $e$, and can at most have one intersection in between them. At most $d(d-1)$ of these edges can be considered simultaneously for different $e_1, \ldots, e_d$ from the $k\lambda \beta$ edges incident to $v$. Therefore, for every edge $e$, we can find at least a set of $(k-1)\lambda \beta d(d-1) \geq (k-(d+1))\lambda \beta $ new edges that in $B \setminus A$. This means that at this point we have already covered $k\lambda \beta (k-(d+1))\lambda \beta$ edges in both $A$ and $B$. One can see that we can continue covering new edges just like in the previous lemma, such that at iteration $l$, the number of covered edges is at least
$$ k(k-(d+1))(k-2(d+1))\ldots (k- l(d+1)) (\lambda \beta)^l,$$
for $l \leq \frac{k-1}{d+1}$. It is easy to see that for $l = \frac{k-1}{d+1}$, we will have $k(k-(d+1))(k-2(d+1))\ldots (k- l(d+1)) (\lambda \beta)^l > 2n\beta$ if $k = \Omega(\log{n})$, which will be a contradiction.
\end{proof}

\section{Figures and Tables}\label{appendixfigures}

\begin{center}

\begin{tabular}{|c|c|c|c|c|c|c|c|c|c|}
  \hline
  $n$ & $m$ & $d$ & $k$ & \#I & Gr & IS  & HEDCS & $\beta$ &  $\lambda$   \\
  \hline
  \hline
  15 & 200 & \multirow{4}{*}{3} & 5 & \multirow{4}{*}{500} & 79.1\% & 87.5\%& 82.3\%  & 5 & 3  \\ 
  30 & 400 &  & 5 &  & 82.6\% & 91.3\%& 84.4\%  & 7 & 4   \\
  100 & 3200 &  & 10 &  & 83.9\% & 96.2\% & 88.2\% & 5 & 2  \\
  300 & 4000 &  & 10 &  & 81.1\% & 92.0\% & 86.3\% & 4 & 2  \\
  \hline
   50 & 800 & \multirow{4}{*}{5} & 6 & \multirow{4}{*}{500} & 76.0\% & 89.1\%& 78.5\%  & 16 & 11   \\
  100 & 2,800 &  & 10 &  & 77.9\% & 92.1\% & 81.9\%  & 16 & 11   \\
  300 & 4,000 &  & 10 &  & 77.8\% & 93.9\% & 87.1\%  & 10 & 5  \\
 500 & 8,000 &  & 16 &  & 79.2\% & 94.3\%& 85.9\% & 10 & 5  \\
 \hline
 500 & 15,000 & \multirow{3}{*}{10} & 16 & \multirow{4}{*}{500} & 71.8\% & 90.6\% & 79.2\%  & 20 & 10    \\
  1,000 & 50,000 &  & 20 &  & 73.8\% & 92.3\% & 81.5\% & 20 & 10   \\
  2,500 & 100,000 &  & 20 &  & 72.2\% & 91.5\% & 80.7\%  & 20 & 10 \\
  5,000 & 200,000 &  & 20 &  & 72.5\% & 90.7\% & 79.8\%  & 20 & 10 \\
  \hline
  1,000 & 5,0000 & 25 & 20 & 100 & 68.2\%  & 87.5\% & 75.6\% & 75 & 50 \\
   2,500 & 100,000 &  & 25 &  & 69.0\% & 87.9\% & 74.3\%  & 75 & 50  \\
   5,000 & 250,000 &  & 30  &  & 67.8\% & 87.3\%  & 75.1\%  & 75 & 50  \\
   10,000 & 500,000 &  & 30 &  & 67.2\% & 86.9\% &  73.7\%  & 75 & 50\\
   \hline
   5,000 & 250,000 & \multirow{4}{*}{50} & 30 & \multirow{4}{*}{100} & 67.4\% & 86.6\% & 74.0\% & 100 & 50  \\
  10,000 & 500,000 & & 30 &  & 68.1\% & 87.1\%& 73.4\%  & 100 & 50  \\
  15,000 & 750,000 &  & 30 &  & 66.9\% & 86.2\% & 73.2\% & 100 & 50   \\
  25,000 & 1,000,000 & & 30 &  & 67.3\% & 86.0\% & 72.8\% & 100 & 50  \\
  
   \hline
\end{tabular}
$ $
\captionof{table}{Comparison on random uniform instances with maximal matching benchmark.} 
\end{center}

\section{Chernoff Bound}

Let $X_1, \ldots, X_n$ be independent random variables taking value
in $[0, 1]$ and $X:= \sum\limits_{i=1}^n X_i$. Then, for any $\delta \in (0, 1)$
$$
    Pr\Big(|X-\mathbb{E}[X]| \geq \delta \mathbb{E}[X]\Big) \geq 2\cdot \exp{\Big(-\frac{\delta^2\mathbb{E}[X]}{3}\Big)}.
$$

\section{Nesbitt's Inequality}\label{nesbitt}

Nesbitt's inequality states that for positive real numbers $a$, $b$ and $c$,

$$ \frac{a}{b+c} + \frac{b}{a+c}+ \frac{c}{a+b} \geq \frac{3}{2},$$

with equality when all the variables are equal. And generally, if $a_1, \ldots, a_n$ are positive real numbers and $s = \sum\limits_{i=1}^n$, then:

$$ \sum\limits_{i=1}^n \frac{a_i}{s-a_i} \geq \frac{n}{n-1},$$

with equality when all the $a_i$ are equal.

\appendix

\end{document}